\documentclass[12pt]{article}
\usepackage[margin = 1.3 in]{geometry}
\usepackage[utf8]{inputenc}
\usepackage{amsfonts}
\usepackage{amsmath}
\usepackage{amssymb}
\usepackage[english]{babel}
\usepackage{amsthm}
\usepackage{tikz}
\usepackage{float}
\usepackage[export]{adjustbox}
\usepackage{subcaption}
\usepackage{bbm}
\usepackage{comment}
\usepackage{natbib}
\usepackage{enumerate}
\usepackage{enumitem}
\usepackage{hyperref}
\usepackage{cleveref}
\setlist{nosep}
\usetikzlibrary{decorations.pathreplacing}
\setcitestyle{authoryear,open={(},close={)}}

\excludecomment{note}

\newtheorem{observation}{Observation}
\newtheorem*{definition}{Definition}

\theoremstyle{definition}
\newtheorem{theorem}{Theorem}
\newtheorem{proposition}{Proposition}
\newtheorem{lemma}{Lemma}
\newtheorem{corollary}{Corollary}

\theoremstyle{remark}
\newtheorem*{remark}{Remark}

\DeclareMathOperator*{\argmax}{\arg\max}

\newcommand{\und}{\underline}
\newcommand{\vst}{\vspace{3mm}}

\usepackage{inputenc}

\DeclareMathOperator{\siminf}{\text{siminf}}
\DeclareMathOperator{\simsup}{\text{simsup}}
\DeclareMathOperator{\gsiminf}{\text{gsiminf}}
\DeclareMathOperator{\gsimsup}{\text{gsimsup}}

\DeclareMathOperator{\e}{\text{E}}

\title{Subjective Complexity Under Uncertainty}
\author{Quitz\'{e} Valenzuela-Stookey\thanks{Departmetn of Economics, UC Berkeley; email: \texttt{quitze@berkeley.edu}. Many thanks to Eddie Dekel, Alessandro Pavan, and especially Marciano Siniscalchi for feedback on this project from its inception. I am also grateful to Edi Karni, Wojciech Olszewski, Francisco Poggi, Ludvig Sinander, Bruno Strulovici, Marie-Louise Vier{\o}, participants at the Transatlantic Theory Workshop 2019, RUD 2020, and seminar participants at Northwestern University.}}
\date{October 13, 2022}

\begin{document}

\maketitle

\begin{abstract}
Complexity of the problem of choosing among uncertain acts is a salient feature of many of the environments in which departures from expected utility theory are observed. I propose and axiomatize a model of choice under uncertainty in which the size of the partition with respect to which an act is measurable arises endogenously as a measure of subjective complexity. I derive a representation of incomplete Simple Bounds preferences in which acts that are complex from the perspective of the decision maker are bracketed by simple acts to which they are related by statewise dominance. The key axioms are motivated by a model of learning from limited data. I then consider choice behavior characterized by a ``cautious completion'' of Simple Bounds preferences, and discuss the relationship between this model and models of ambiguity aversion. I develop general comparative statics results, and explore applications to portfolio choice, contracting, and insurance choice.
\end{abstract}

Acts with uncertain outcomes are complicated, potentially infinite-dimensional, objects. As early as \cite{von1944game} and \cite{aumann1962utility}, it was recognized that even completeness of preferences over such acts may be a lot to ask. The complexity of choosing between uncertain acts has been evoked as an explanation for many empirical observations at odds with the predictions of the subjective expected utility model, from ambiguity aversion (e.g. \cite{gilboa1989maxmin}) to the so-called ``equity premium puzzle'' of \cite{mehra1985equity}.

This paper studies a choice-based notion of complexity in choice under uncertainty. The analysis can be divided into two steps. I first identify a set of \textit{well-understood} acts, and then characterize the (incomplete) preferences of a decision maker who uses well-understood acts as a tool for making difficult comparisons. Say that an act $f$ (a map from the state space to a set of outcomes) is well understood if it has a certainty equivalent, i.e. a constant act to which it is indifferent.\footnote{The idea is that if a decision maker understands an act well, they should be able to calibrate it's value precisely in relation to constant acts. Conversely, a certainty equivalent can be used as a simple proxy for $f$ when comparing $f$ to other acts, rendering such comparisons relatively easy.} The first step in the analysis is to characterize the set of well-understood acts. 

Based on assumptions regarding the ability of the decision maker to make comparisons between constant acts various types of non-constant acts, I derive two nested characterizations of the set of well-understood acts. First, I show that \textit{i}) whether or not an act is well-understood depends only on its partition, i.e. the coarsest partition of the state space with respect to which the induced payoffs are measurable, and \textit{ii)} that acts with coarser partitions are easier to understand. Under fairly weak conditions, these two properties fully characterize the set of well-understood acts (\Cref{prop3.4}). Under an additional condition, I derive a second stronger characterization of well-understood acts: there exists an integer $N$, determined endogenously by individual choice behavior, such that an act is well-understood if and only if it takes no more than $N$ distinct values (\Cref{thm:simplicity_char}). This parsimonious and tractable characterization of the well-understood acts forms the basis for the subsequent analysis. 

I impose the standard EU axioms, with minor modifications, on the set of well-understood acts, leading to the standard expected utility characterization of preferences over such acts. The second step of the analysis concerns the ability of the decision maker to make comparisons involving acts that are not well-understood. A single axiom characterizes the minimal (in the set-inclusion sense on binary relations) extension of preferences from the well-understood set to the space of all acts, consistent with monotonicity (respecting state-wise dominance) and transitivity. This characterization is the central contribution of the paper. Under this condition, preferences admit a representation whereby the value of any act $f$ can be bounded above and below by the worst (best) well-understood act that state-wise dominates (is state-wise dominated by) $f$. If at least one of the acts $f$ and $g$ are not well understood, and the acts are not related by state-wise dominance, then $f$ is preferred to $g$ if and only if the ``simple greatest lower bound'' of $f$ is preferred to the ``simple least upper-bound'' of $g$. The precise form that this characterization takes depends on the characterization of the well-understood acts. If we use the weaker characterization of the well-understood set (\Cref{prop3.4}) we obtain a more general characterization of preferences (\Cref{thm:general_rep}). Under the stronger well-understood characterization (\Cref{thm:simplicity_char}) we obtain a more tractable, albeit restrictive, preference characterization (\Cref{thm3.1}). I focus primarily on the latter result, and refer to these as \textit{Simple Bounds} preferences. Relative to standard EU preferences, Simple Bounds preferences have only one additional parameter, $N$, which is pinned-down by choice behavior.

The notion of complexity that arises endogenously from the assumptions on choice behavior, the number of distinct values in the range, can be interpreted as reflecting the difficulty people have in contemplating the relative likelihoods and aggregating payoffs across multiple events. To formalize this intuition, I study a procedural model of decision making in which a frequentist decision maker uses a finite set of i.i.d. observations to learn about the unknown distribution of the state. I show that the axioms allowing for the characterization of complexity in terms of partition size are implied by natural assumptions on the procedural model of choice. 

I consider two models of complete preferences derived from Simple Bounds. These models complement the analysis of incomplete preferences by exploring how the DM handles their lack of understanding. Preferences are \textit{Cautious} if the DM evaluates an act $f$ according to its simple lower bound, i.e. the best act that has a certainty equivalent and is dominated statewise by $f$. Similarly preferences are \textit{Reckless} if acts are evaluated according to their simple upper bound, defined analogously. I axiomatize these two preferences, relate them to the DM's attitude towards ambiguity, and compare them to existing models of ambiguity aversion. 

I provide some general comparative statics results, building on \cite{tian2015optimal} and \cite{tian2016monotone}, which are useful in applications. Finally, I investigate three applications of the model. First, I study consumption savings problems and equilibrium asset prices. Cautious decision makers will save more, and allocate a greater portion of their savings to a safe rather than a risky asset, compared to fully rational individuals. In equilibrium these biases lead to higher prices for safe relative to risky assets, as observed in the ``equity premium puzzle'' of \cite{mehra1985equity}. Second, I study the choice of insurance plans. I show that the cautious model rationalizes many ``behavioral'' phenomena identified in the empirical literature, such as over-sensitivity to deductibles and coverage rates. It also helps explain the correlation, documented by \cite{bhargava2017choose}, between specific forms of dominated plan choice and both the degree of health risk and level of education. This application provides a good illustration of some of the general comparative statics results presented in Appendix \ref{app:compstat}. Finally, I examine principal-agent contracting with a complexity constrained agent, identifying general features of optimal contracts.

Section \ref{sec2} introduces the setting. Section \ref{sec3} presents the characterization of subjective complexity. Section \ref{sec:simplebounds} presents the Simple Bounds characterization. Section \ref{sec:learning} explores the procedural learning model which motivates the key axioms. Section \ref{sec4} discusses complete preferences derived from Simple Bounds. \Cref{sec:applications} presents the applications. Related literature is discussed in Section \ref{sec10}. Omitted proofs are presented in \Cref{app:omittedproofs}. General comparative statics results and further discussion are contained in the Online Appendix (\Cref{sec:online}). 

\section{The model}\label{sec2} 

The framework is that of \cite{anscombe1963definition}. The decision maker is characterized by a binary relation $\succsim$ over acts (the preference relation). The strict preference relation and indifference relation are defined as usual.\footnote{I take as primitive the reflexive and transitive relation, interpreted as weak preferences. An alternative approach, as in \cite{galaabaatar2013subjective}, would be to take as primitive a transitive and irreflexive strict partial order. I find my approach more convenient because I make use directly of the existence of certainty equivalents and the transitivity of the weak preference relation.} Further notation is as follows:
\begin{itemize}
    \item $Z$: the set of \textit{outcomes}.
    \item $L$: the set of vN-M \textit{lotteries} (finite support distributions) over $Z$.
    \item $\Omega$: the \textit{state space}, endowed with an algebra $\Sigma$ of \textit{events}. 
    \item $F_c$: the set of constant acts.
    \item $F$: the set of finite valued acts; $\Sigma$-measurable $f: \Omega \rightarrow L$ such that $|f(\Omega)| < \infty$.
\end{itemize}
The preference relation is assumed to satisfy the following basic conditions (see e.g. \cite{gilboa2010objective}). 
\vspace{1.5mm}

\noindent \textbf{Basic Conditions}

\vspace{1mm}
\noindent PREORDER: $\succsim$ is reflexive and transitive. 

\vspace{1mm}
\noindent MONOTONICITY: For every $f,g \in F$, $f(\omega) \succsim g(\omega) \ \forall \ \omega \in \Omega$ implies $f \succsim g$. 

\vspace{1mm}
\noindent ARCHIMEDEAN CONTINUITY: For all $f,g,h \in F$, the sets $\{\lambda \in [0,1]: \lambda f + (1-\lambda) g \succsim h \}$, $\{\lambda \in [0,1] : h \succsim \lambda f + (1-\lambda) g \}$ are closed in $[0,1]$.

\vspace{1mm}
\noindent NONTRIVIALITY: There exist $f,g \in F$ such that $f \succ g$.

\vspace{1mm}
\noindent C-COMPLETENESS: $\succsim$ is complete on $F_c$

\vspace{1mm}
\noindent WEAK C-INDEPENDENCE: Let $c_1,c_2$ be constant acts. Then for any act $f$ and $\alpha \in (0,1)$, $f \succsim (\precsim) c_1$ if and only if $\alpha f + (1- \alpha)c_2 \succsim (\precsim) \alpha c_1  + (1- \alpha)c_2$.

\vspace{1.5mm}
The key difference between the Basic Conditions and the corresponding subset of the standard SEU axioms is of course the lack of completeness. Transitivity is preserved as a basic tenet of rational preferences. For a detailed discussion of the relationship between transitive but incomplete preferences and intransitive choice see \cite{mandler2005incomplete}. Weak C-Independence is a weakening of the common C-Independence assumption, as it imposes that two of the acts involved be constant.

It should be noted that monotonicity rules out certain types of behavior that may be thought of ``complexity averse''. Generally speaking, monotonicity precludes behaviors that arise if the decision maker finds it difficult to translate the description of an act into a mapping from states to lotteries. For example, \cite{ellis2017correlation} study a decision maker whose misperception of correlation leads to violations of monotonicity. In my model, as discussed in detail below, the decision maker understands the state space, as well as the mapping from states to lotteries defined by each act. The difficulty lies in aggregating across states to make comparisons between acts. 

The Basic Conditions imply that there is a (complete) expected utility representation of preferences on constant acts, which are identified with the corresponding lottery. This is in keeping with the literature on ambiguity, stemming from \cite{gilboa1989maxmin}, which distinguishes between risky lotteries, which the agent has no difficulty evaluating, and uncertain acts, over which the DM may exhibit non-EU behavior.\footnote{There is of course empirical evidence that choice over risky lotteries may also be influenced by complexity constraints (see for example \cite{neilson1992some}). More recently, models have been developed which incorporate both ambiguity aversion and non-EU preferences over risky lotteries (see for example \cite{dean2017allais}). The techniques developed in the current paper for studying choice under uncertainty can be applied to risky choice as well. Studying choice behavior when both constraints are present, along the lines of \cite{dean2017allais}, is an intriguing direction for future work.}

Let $v$ represent preferences on $L$. Say that an act $f$ is $N$-simple if $|v\circ f(\Omega)| \leq N$, and $N$-complex otherwise. Let $F_N$ be the set of $N$-simple acts. I will refer to the \emph{coarsest} partition of $\Omega$ with respect to which $v\circ f$ is measurable as \textit{f's partition} (so an act is $N$-simple iff its partition has no more than $N$ elements). I will say that $f$ is \textit{N-simple on E} if $|v\circ f(E)| \leq N$. I make no behavioral assumptions regarding $N$-simple acts directly, but they will turn out to be closely related to perceived complexity. Intuitively, the fact that what ends up mattering is the range of $v \circ f$, rather than $f$, reflects that the DM does not find it difficult to evaluate individual consequences, but rather to understand the mapping from states to consequences. Moreover, distinguishing between acts with the same utility image is incompatible with Monotonicity; if $v\circ f = v\circ g$ then Monotonicity implies $f \sim g$. Monotonicity, i.e. respecting statewise dominance, I take as a basic tenet of rationality. In many applications, such as the portfolio choice and insurance applications considered here, the DM has a strict ordering on $v \circ f(\Omega)$, in which case the distinction between the partition of $f$ and $v \circ f$ disappears.

The size of an act's partition is a common measure of complexity in both the theoretical and empirical literature (see Section \ref{sec10}). This notion captures the idea that the DM has trouble \textit{i}) forming beliefs about many events and \textit{ii}) combining a large number of potential outcomes to understand an act's value. Unlike most existing papers however, I do not assume that partition size represents subjective complexity. Rather, I show that this measure arises endogenously as a characterization of well-understood comparisons. 

\section{Characterizing understanding}\label{sec3}

In order to characterize complexity as perceived by the DM, I require a choice-based measure of the types of comparisons that the DM finds it easy to make. Monotonicity gives us some information along these lines, as it says that the DM is always able to rank acts that are ordered by statewise dominance. The easiest acts to compare to an arbitrary act $f$, aside from those related to $f$ by statewise dominance, are, intuitively, constant acts. The more constant acts the DM is able to rank against $f$ the better $f$ is understood. Motivated by this intuition, I will begin by making assumptions about comparisons to constant acts, and in particular about the set of acts for which there exists a certainty equivalent, i.e. a constant act to which it is indifferent. Let $F_{CE}$ be the set of such acts, which I call \textit{well-understood}.\footnote{It is intuitive that acts which are easy to understand have certainty equivalents. Conversely, I will assume further on that Independence, another intuitive property of easy to understand acts, is satisfied on $F_{CE}$.}

A bet on an event $E$ is a binary act that has a strictly better consequence on $E$ then on the complement of $E$. Given any acts $f,g$, let $f_Eg$ be the binary act which is equal to $f$ on $E$ and $g$ elsewhere. The following axioms characterize the set of acts that are well-understood by the decision maker. 

\vst
\noindent \textbf{Simplicity Conditions (part 1).} Let $\tau$ be $f$'s partition, with typical elements $T,T'$.

\vspace{1mm}
\noindent A0. For any event $E \in \Sigma$ there exists a bet on $E$ that is well-understood.

\vspace{1mm}
\noindent A1. If $f$ is well-understood then for any partition $\tau'$ coarser than $\tau$, there exists a well-understood act with partition $\tau'$. 

\vspace{1mm}
\noindent A2. Let $f$ be a well-understood act. Let $g$ measurable with respect to $f$'s partition, and for which there exists a constant act $a \in f(\Omega)$ such that for all $\omega \in \Omega$ either $g(\omega) = f(\omega)$ or $g(\omega) = a$. Then if $g$ is well-understood, $\alpha f + (1-\alpha)g$ also is well-understood for all $\alpha \in (0,1)$.

\vst
The Simplicity Conditions (parts 1 and 2) are further motivated by a procedural model of learning, presented in Section \ref{sec:learning}. Axioms A0-A2 (and A3 below) are stated using the ``there exists an act'' qualifier, rather than ``for any act'' for two reasons. First, the learning model, which helps clarify the content of the Simplicity Conditions, implies the ``there exists'' version of the axioms, but not the ``for any''. Second, to state the axioms ``for any act'' would essential be to assume that complexity is a property of an act's partition. Instead, I show exactly which assumptions characterize this type of complexity. 

Intuitively, A0 says that there are no events which the DM does not understand. This allows us to identify a subjective probability measure on the state space that characterizes the DM's preferences over binary acts. This is in contrast to models such as \cite{epstein2001subjective} in which there are events that are inherently difficult to understand. In \cite{epstein2001subjective} it is possible for probabilistic sophistication, as defined in \cite{machina1992more}, to be violated whenever an act is not measurable with respect to a set of subjectively unambiguous events. A0 embodies the idea that complexity is a property of comparisons between acts, rather than an inherent difficulty with understanding certain events. 

Roughly, A1 says that acts with coarser partitions are easier to understand. For example, if we merge two cells in the partition of $f$ and replace the outcome on the merged cell with the conditional expectation of $f$ over the two cells, we might expect that the act has become easier to evaluate. (In the procedural model of choice of \Cref{sec:learning} I show that it is precisely this type of conditional-mean modification which makes an act easier to compare.)

A2 is similar to the Certainty Independence assumption of \cite{gilboa1989maxmin}. Certainty Independence would imply that if $f$ is well-understood and $c$ is a constant act, then $\alpha f + (1-\alpha)c$ would also be well-understood. A2 generalizes this conclusion to mixtures with non-constant acts $g$ that are well-understood, provided $g$ has the same partition and range as $f$, and is constant on the set of states on which it differs from $f$. The usual arguments in favor of Certainty Independence, based on the fact that mixing with a constant act does not help to hedge against uncertainty, therefore apply.

Under the Basic Conditions, the Simplicity Conditions characterize the set of well-understood acts. This characterization can be separated into two parts. First, A0-A2 imply that \textit{i}) if any act with a given partition is well-understood then so are all such acts, and \textit{ii}) with the refinement order on the space of partitions, the set of partitions for which measurable acts are well-understood is a lower-set.\footnote{Recall the distinction between an act ``with partition $\tau$'' and a $\tau$-measurable act. The former means $\tau$ is the \textit{coarsest} partition w.r.t which the act is measurable.} 

\begin{proposition}\label{prop3.4}
Under A0-A2 and the Basic Conditions, if some act with partition $\tau$ is well-understood then so are all $\tau$-measurable acts.
\end{proposition}

Conceptually, Proposition \ref{prop3.4} implies that complexity is a property of an act's partition. One way to interpret Proposition \ref{prop3.4} is as follows. Along any sequence of increasingly fine partitions, there is a last partition such that all acts with this partition have certainty equivalents. This alone is an interesting and potentially useful characterization of understanding. A representation based on \Cref{prop3.4} is presented in \Cref{sec:generalized_representation}. If we look only at acts with partitions ordered by refinement, then Proposition \ref{prop3.4} tells us all we need to know. In many applications, however, it will be helpful to have a bit more structure on the model of choice. This structure is delivered by part 2 of the Simplicity Conditions. Say that an event $E$ is \textit{non-null} if there exist $c^1,c^2,c^3 \in F_c$ such that $c^2_Ec^1 \succ c^3_Ec^1$, and \textit{null} otherwise (the empty set is null, by reflexivity of $\succsim$).\footnote{As we will see, we will be able to identify a probability $P$ which the DM uses to evaluate well-understood acts, and null events will be zero measure events under $P$.}

\vst
\noindent \textbf{Simplicity Conditions (part 2).} Let $\tau$ be $f$'s partition, with typical elements $T,T'$.

\vspace{1mm}
\noindent A3. If $f$ is well-understood then for any non-null $T\neq T'$ in $\tau$, and any $2$-element partition $\tau'$ of $T \cup T'$, there exists a well-understood act with partition $\tau' \cup (\tau\setminus\{T,T'\})$.

\vspace{1mm}
\noindent A4. If $f$ is well-understood and $E$ is null then $g_Ef \sim f$ for all $g \in F$.

\vst
The key axiom is A3, with A4 playing more of a technical role. A3 says that if $f$ is well-understood, then there are acts that are ``close'' to $f$ that are also well understood. Here ``close'' has two meanings. First, the partitions of the acts can only differ on two cells. Second, since A3 only implies that there exists \textit{some} well understood act with partition $\tau' \cup (\tau\setminus\{T,T'\})$, the values of the new act can be close to those of $f$. If the new act is ``nearly constant'' on $\tau'$, then A3 is conceptually close to A1. This connection will be formalized in the learning model of Section \ref{sec:learning}. As we will see, the Basic Conditions, A0-A2, and S-Independence, discussed below, imply that there is an expected utility representation of preferences over well-understood acts. In light of this, A4 simply says that the decision maker is not artificially confused by changes to a well-understood act that occur with zero probability. 

Relative to the characterization in \Cref{prop3.4}, assumptions imply that along any sequence of refined partitions, the last partitions with the ``certainty equivalent property'' have the same number of non-null elements.

\begin{theorem}\label{thm:simplicity_char}
Assume $\succsim$ satisfies the Basic Conditions and Simplicity Conditions. Then there exists an $N \in (\mathbbm{N} \backslash \{1\}) \cup \{\infty\}$ such that an act is well-understood iff its partition has at most $N$ non-null elements. 
\end{theorem}

Another way of stating the characterization in Theorem \ref{thm:simplicity_char} is that an act is well-understood iff it is $N$-simple on all but a null event. This theorem provides a choice-based foundation for the intuitive, and widely used, measure of complexity as partition size. 

The extension from Proposition \ref{prop3.4} to Theorem \ref{thm:simplicity_char} is non-trivial. It is not the case that Theorem \ref{thm:simplicity_char} would follow immediately if ``there exists'' was replaced with ``for any'' in axioms A1-A3. The key step in the proof is a novel algorithm for moving between partitions of the state-space while preserving the ``certainty equivalent property'' (\Cref{lem3.1}). 

\subsection{\texorpdfstring{\Cref{thm:simplicity_char}}{} discussion}\label{sec:thm1_evidence}

Consider the refinement partial order on the set of partitions of $\Omega$. We can represent the set of all partitions as a directed acyclic graph, with the coarsest partition of $\Omega$ as the root node. As we move along each path in this graph (i.e. each chain in the refinement partial order), partitions become increasingly fine. \Cref{prop3.4} says that along every path, there is a single node at which partitions lose the well-understood property. \Cref{thm:simplicity_char} strengthens this conclusion by saying that by saying that all of the paths on which the well-understood property is preserved have the same length. 

While the restriction that the well-understood property depends only on the number of elements is strong, it is supported by a large body of experimental and observational evidence relating partition size to perceived complexity. \cite{moffatt2015heterogeneity} and \cite{sonsino2002complexity} find experimental evidence that subjects undervalue lotteries with larger supports. \cite{bernheim2019direct} argue that aversion to large supports helps explain experimental data that is otherwise inconsistent with both expected utility and cumulative prospect theories. Additional experimental evidence, relating support size to ambiguity aversion, is discussed in \Cref{sec:experimental_ambiguity}

\section{Main results: Simple Bounds representation}\label{sec:simplebounds}

I turn now to representing preferences. The goal is to stick as close as possible in this regard to the SEU model. However some modifications to the standard SEU axioms must be made to maintain the spirit of the previous axioms. In particular, the usual independence assumption must be modified. In general, the mixture of two $N$-simple acts will not be $N$-simple. The standard independence axiom would thus expand the set of well-understood acts; it would imply $N = \infty$. The modified axiom, S-Independence, eliminates this concern. Moreover, it addresses some of the usual critiques of the independence assumption. For one, it applies only when all acts involved are well-understood. Moreover, it only applies to mixtures that do not result in too complex an act, in the sense of partition fineness. 

\vspace{1.5mm}
\noindent S-INDEPENDENCE: Let $f,g,h$ be well-understood acts. For any $\alpha \in (0,1)$, if $\alpha f + (1-\alpha) h$ and $\alpha g + (1-\alpha) h$ are both well understood, then $f \succsim g$ if and only if $\alpha f + (1-\alpha) h \succsim \alpha g + (1-\alpha) h$.

\vspace{1.5mm}
Finally I make an assumption on how comparisons involving complex acts can be made. I write $f \geq_E g$ if  $f(\omega) \succsim g(\omega) \ \forall \ \omega \in E$. If there exists a null set $E$ such that $f \geq_{\Omega\setminus E} g$, write $f \geq^0 g$.  For simplicity, I write $\geq$ rather than $\geq_{\Omega}$ to represent statewise dominance.

\vspace{1.5mm}
\noindent UNIFORM COMPARABILITY: For any $f,g \in F$ such that $\neg(f \geq g)$, $f \succsim g$ holds iff there exist well-understood acts $h,k$ such that $f \geq h \succsim k \geq g$.
\vspace{1.5mm}

Of course, if $f$ and $g$ are well understood then $f=h$ and $g=k$. In general, it seems reasonable to assume that comparisons according to statewise dominance can be made, even if the acts in question are complex. Uniform Comparability says that statewise dominance is the \textit{only} way to compare complex acts. Thus Uniform Comparability implies the minimal extension of preferences beyond well-understood acts. The ``only if'' direction formalizes the idea that complexity causes incompleteness of preferences. 

Uniform Comparability leads directly to a representation of preferences in which complicated acts are bracketed by well-understood acts, and $f \succsim g$ if and only if the best well-understood act dominated statwise by $f$ is preferred to the worst well-understood act that statewise dominates $g$. A convenient way to describe this bracketing is through simple upper and lower bounds. 

\begin{definition}
For any $f \in F$ and $N \in \mathbb{N}$, denote by \textbf{$\simsup_{\textbf{N,f}}$} the set of acts $h$ satisfying: 1) $h \in F_N$, 2) $h \geq^0 f$, and 3) there is no $k$ satisfying 1 and 2 such that $h \succ k$.\footnote{Requiring only $h \geq^0 f$, rather than $h \geq f$, in condition 2 is a technicality. For most applications, we can replace this condition with $h \geq f$, as Proposition $\ref{prop:sim_char}$ shows.}
\end{definition}

Define $\siminf_{\textit{N,f}}$ analogously. Note that for any $N$ such that all acts in $F_N$ are well-understood, the DM is indifferent between all acts in $\simsup_{N,f}$ (similarly for $\siminf_{N,f}$). When there is no risk of confusion, I will therefore abuse notation and write as if $\siminf$ and $\simsup$ are single valued (for example, $f \succsim \siminf_{N,f}$ even when $\siminf_{N,f}$ may contain multiple acts). It is not obvious that such bounds exist, i.e. that $\siminf_{N,f}$ and $\simsup_{N,g}$ are non-empty. In fact, the axioms stated will imply that both are non-empty. Conversely, \Cref{thm3.2} shows that existence of these bounds places no additional restrictions on the model parameters. 

For a lottery $l$, I write $E_l u =  E_l[u(z)]$, i.e. the expected utility given distribution $l$ over outcomes. Although $\siminf_{N,f}$ may not be single valued, I abuse notation and write $\int_{\Omega}E_{\siminf_{N,f}(\omega)}u dP(\omega)$.

\begin{definition}
Preference $\succsim$ has a \textbf{Simple Bounds} representation if there exists an integer $N$, probability $P$ on $\Sigma$, and a non-constant function $u: Z \rightarrow \mathbb{R}$ such that, for every $f,g \in F$, $f \succsim g$ if and only if at least one of the following holds:
\renewcommand{\theenumi}{\roman{enumi}}
    \begin{enumerate}
        \item $f \geq g$.
        \item
        \begin{equation*}
            \int_{\Omega} E_{\siminf_{N,f}(\omega)}u \ dP(\omega) \geq \int_{\Omega} E_{\simsup_{N,g}(\omega)}u \ dP(\omega) 
        \end{equation*}
    \end{enumerate}
\end{definition}
\noindent I refer to such preferences as Simple Bounds preferences.

\begin{theorem}\label{thm3.1}
The following statements are equivalent:
\renewcommand{\theenumi}{\roman{enumi}}
\begin{enumerate}
    \item $\succsim$ satisfies the \textbf{Basic Conditions}, \textbf{Simplicity Conditions}, \textbf{S-Independence}, and \textbf{Uniform Comparability}.
    \item $\succsim$ has a Simple Bounds representation, with parameters $P,u,N$. Moreover, $P$ is unique, $u$ is unique up to positive affine transformations, and for all $f \in F$, $\siminf_{N,f}$ and $\simsup_{N,f}$ are non-empty.
\end{enumerate}
\end{theorem}

Condition \textit{i.} in the definition of a Simple Bounds representation requires $f\geq g$, rather than $f \geq^0 g$. This is natural if we think that it is harder to identify null events when comparing complex acts. Alternatively, we could assume that preferences obey $\geq^0$ dominance, and make the obvious modification to Uniform Comparability, to replace  \textit{i.} with $f \geq^0 g$. 

For applications, it is generally without loss to assume that $\simsup_{N,f} \geq f \geq \siminf_{N,f}$. Roughly speaking, violations of statewise dominance only occur if $f$ has isolated outliers in the support of $P$. No continuity is required. 

\begin{proposition}\label{prop:sim_char}
$f \geq siminf_{N,f}$ iff $f^{-1}(A)$ is non-null for every open neighborhood of $\inf E_f u(\Omega)$. Similarly, $\simsup_{N,f} \geq f$ iff $f^{-1}(A)$ is non-null for every open neighborhood of $\sup E_f u(\Omega)$.
\end{proposition}

The conditions of Proposition \ref{prop:sim_char} are met, for example, if $P$ has full support on $\Omega \subseteq \mathbb{R}$ and $E_fu$ is continuous, or if its partition has no null elements, or it is a convex combination of any acts with these properties. Even if the conditions for statewise dominance fail, it is easy to see where violations will occur. The ``Lebesgue approach'', discussed below, helps clarify this point. 

Theorem \ref{thm3.1} includes the conclusion that $\siminf_{N,f}$ and $\simsup_{N,f}$ are non-empty for all $f$. A natural concern from a modeling perspective is that this may impose constraints on the other parameters of the representation. This would be the case if existence failed for some specification of $u, P$ or $N$. The following result states that this is not the case. Let $B(\Omega)$  be the set of bounded measurable functions on $\Omega$, and $B_N(\Omega)$ be the set of $N$-simple measurable functions on $\Omega$. Abusing notation, define
\begin{equation*}
\siminf_{N,w,P} \equiv \argmax_{\{b \in B_N(\Omega): b \leq^0 w\}} \int_{\Omega} w(\omega) dP(\omega)
\end{equation*}

\begin{proposition}\label{thm3.2}
For any $w \in B(\Omega)$, $\siminf_{N,w,P}$ and $\simsup_{N,w,P}$ are non-empty for all $N,P$. 
\end{proposition}
A full discussion of \Cref{thm3.2} can be found in \Cref{sec:existenceproofs}. The proof is instructive, as it makes use of a dual approach to the problem of finding simple upper and lower bounds; rather than look for functions defined by partitions on $\Omega$, I define a dual problem in terms of partitions of $w(\Omega)$. Using this dual ``Lebesgue approach'' finding the $\siminf_{N,w,p}$ and $\simsup_{N,w,P}$ for an arbitrary act $f$ can be reduced to the problem of finding a the $\siminf$ and $\simsup$ for the identity function on the unit interval, with a suitably defined distribution $\hat{P}$.

This dual approach is useful for comparative statics (Appendix \ref{app:compstat}). Moreover, it is easy to use this approach to show that finding $\simsup_{N,w,P}$ and $\siminf_{N,w,P}$ functions can be reduced to the computationally easy problem of finding the maximum-weight path in a suitably defined directed acyclic graph. It also implies a convenient structure of the sets $\siminf_{N,w,P}$ and $\simsup_{N,w,p}$.

\begin{lemma}\label{lem:lattice}
$\siminf_{N,w,P}$ and $\simsup_{N,w,P}$ are lattices. 
\end{lemma}

The partial order with respect to which \Cref{lem:lattice} holds is discussed in \Cref{app:compstat}. For $\Omega = [0,1]$, the result follows from \Cref{prop202}. The extension to more general spaces follows from the dual approach of \Cref{sec3.1}.

\subsection{Theorem \ref{thm3.1} discussion}\label{sec:thm1disc}

Simple Bounds preferences characterize a decision maker who 1) understands some set of acts well, and 2) uses the acts they understand well to bracket those they do not, allowing them to make some comparisons involving poorly understood acts. The characterization provided by \Cref{thm:simplicity_char} of which acts are well understood is sharp: the well-understood property is defined by a cut-off in the number of elements of an act's partition. This does not mean, however, that perturbing an act slightly to increase the size of its partition will cause a discontinuous change in choice behavior. If act $f'$ takes values very close to those of $f$ (uniformly over $\Omega$), where the latter is well understood while the former is not, the simple lower and upper bounds of $f'$ will be almost identical to $f$. Thus the sharp characterization of what is well understood has smoother implications for choice. This is in contrast to models in which complexity enters the representation as an additive cost, such as \cite{puri2020simplicity}. 

A natural question is how the incomplete preferences characterized above relate to those of \cite{bewley2002knightian}.\footnote{\cite{bewley2002knightian} in fact gives a representation of the strict preference. The exact unanimity representation discussed here is due to \cite{gilboa2010objective}.} So called ``Bewley preferences'' have the following representation: there exists a non-empty, closed, and convex set $C^*$ of probabilities on $\Sigma$ and a non-constant function $u: Z \rightarrow \mathbb{R}$ such that $f \succsim g$ if and only if
\begin{equation*}
    \int_{\Omega} E_{f(\omega)}u dp(\omega) \geq \int_{\Omega} E_{g(\omega)}u dp(\omega) \ \ \forall \ p \in C^*.
\end{equation*}

There is no simple relationship between the two representations. The Bewley preferences satisfy the usual independence axiom. However it is easy to see that the representation in Theorem \ref{thm3.1} does not satisfy independence. If incompleteness \`{a} la Bewley is interpreted as reflecting complexity then it must be that mixtures do not increase the complexity of a comparison. For example, let the state space be $[0,1]$ and for any act $k$ identify $u(k(\omega))$ with $k(\omega)$. Let $f$ and $g$ be binary acts, with $f =  0,  g = 1$ on $[0,1/2)$ and $f = 1, g = 0$ on $(1/2,1]$. In the Bewley model, if $f$ and $g$ are comparable then so are $1/2 f + 1/2 h$ and $1/2 g + 1/2 h$ for $h(\omega) = 10 \omega^2$. Decision makers may find the later comparison, which involves acts with a greater range of values and larger partitions, more difficult. Relaxing Independence and explicitly modeling complexity allows for a model in which mixtures can increase complexity and lead to incomparability. 

\subsection{A generalization}\label{sec:generalized_representation}

It is worth noting that \Cref{thm3.1} is a special case of a more general characterization, in which we impose part 1 of the Simplicity Conditions, but not part 2. I focus on the Simple Bounds representation as defined above, in which act complexity is a function of the cardianlity of the partition, because it is tractable in applications, and consistent with empirical evidence (see \Cref{sec:thm1_evidence} and \Cref{sec:experimental_ambiguity}). 

Endow the space of partitions of $\Omega$ with the refinement partial order: write $\tau''R\tau'$ if $\tau''$ is a refinement of $\tau'$. A set $\mathcal{T}$ of partitions of $\Omega$ is \textit{downward closed} if $\tau'' \in \mathcal{T}$ and $\tau''R\tau'$ implies $\tau' \in \mathcal{T}$. For any $f \in F$, let $\gsiminf_{\mathcal{T},f}$ be the set of undominated acts among all $g$ such that $g \leq^0 f$ and $g$ measurable with respect to a partition in $\mathcal{T}$. 

\begin{definition}
Preference $\succsim$ has a \textbf{Generalized Simple Bounds} representation if there exists a downward directed set $\mathcal{T}$, probability $P$ on $\Sigma$, and a non-constant function $u: Z \rightarrow \mathbb{R}$ such that, for every $f,g \in F$, $f \succsim g$ if and only if at least one of the following holds:
\renewcommand{\theenumi}{\roman{enumi}}
    \begin{enumerate}
        \item $f \geq g$.
        \item
        \begin{equation*}
            \int_{\Omega} E_{\gsiminf_{\mathcal{T},f}(\omega)}u \ dP(\omega) \geq \int_{\Omega} E_{\gsimsup_{\mathcal{T},g}(\omega)}u \ dP(\omega) 
        \end{equation*}
    \end{enumerate}
\end{definition}
\noindent I refer to such preferences as Simple Bounds preferences.

\begin{theorem}\label{thm:general_rep}
The following statements are equivalent:
\renewcommand{\theenumi}{\roman{enumi}}
\begin{enumerate}
    \item $\succsim$ satisfies the \textbf{Basic Conditions}, \textbf{A0-A2}, \textbf{S-Independence}, and \textbf{Uniform Comparability}.
    \item $\succsim$ has a Generalized Simple Bounds representation, with parameters $P,u,\mathcal{T}$. Moreover, $P$ is unique, $u$ is unique up to positive affine transformations, and for all $f \in F$, $\gsiminf_{\mathcal{T},f}$ and $\gsimsup_{\mathcal{T},f}$ are non-empty.
\end{enumerate}
\end{theorem}

\section{Learning motivation for the simplicity conditions}\label{sec:learning}

I present here a simple procedural model of a decision maker using data to inform their choice between uncertain acts. The purpose of this section is twofold. First, it demonstrates that the Simplicity Conditions are satisfied under natural assumptions on the learning model. Second, it formalizes the intuition that acts with coarser partitions are easier to understand. 

The setting is that of a standard frequentist inference problem. A decision maker is endowed with a dataset $\chi$ of $K$ i.i.d. draws from an unknown distribution $P$.\footnote{This sampling procedure can be interpreted literally as sampling from an unknown distribution. Alternatively, we can think of it as a reduced form model of contemplation in which the DM accesses a latent belief $P$, similar in spirit to drift-diffusion models of cognition.} They must compare an arbitrary simple real valued act $f$ and constant act $c$. These can be thought of as the utility images of acts with outcomes in an arbitrary space. The DM uses their data to estimate the expected value of $f$, as well as the risk arising from sampling uncertainty. Let the empirical distribution of the sample be $\hat{P}_{\chi}$. The empirical expectation of $f$ is
\begin{equation}
    \hat{\e}_{\chi}[f] = \sum_{i=1}^N f(T_i)\dfrac{1}{K}\sum_{x\in \chi} \mathbbm{1}\{x \in T_i \}.
\end{equation}
The error due to sampling uncertainty is $\varepsilon_{f}(\chi) := \hat{\e}_{\chi}[f] - \e[f]$. Denote the true distribution of $\varepsilon_f(\chi)$ across different samples by $G_f$. $G_f$ is unknown to the DM, as it depends on the unknown distribution $P$. Instead, the DM uses an estimate $\hat{G}_f$ of $G_f$ when making decisions. I assume that the DM uses the bootstrap to estimate $\hat{G}_f$.\footnote{The bootstrap procedure for estimating $G_f$ is as follows. Draw a sample $\bar{\chi}$ of size $K$ from the empirical distribution (i.e. sample the data with replacement) and calculate $\bar{\varepsilon}_f(\bar{\chi}) := \hat{E}_{\bar{\chi}}[f] - \hat{E}_{\chi}[f]$. Do this repeatedly and use the resulting empirical distribution of $\bar{\varepsilon}_f(\bar{\chi})$ as the estimate of $G_f$. The idea is to treat the empirical distribution $\hat{P}_{\chi}$ as if it were the true distribution, for the purposes of estimating $G_f$. Since $n^{-1/2}\varepsilon_f$ is an asymptotically pivotal statistic, the bootstrap estimator provides an asymptotic refinement of the normal approximation, and therefore performs better in finite samples (see \cite{horowitz2001bootstrap}). Thus it is reasonable for the DM to use the bootstrap to calculate $\hat{G}_f$.} To complete the description of the environment, assume that the DM uses a decision rule that maps $\hat{\e}_{\chi}(f)$ and $\hat{G}$ to a ranking between $f$ and $c$. I also allow the DM to declare that $f$ and $c$ are un-rankable. This is a standard problem of statistical inference. I will show that any protocol that is consistent, in a weak sense, with second order stochastic dominance rankings of the error distributions will satisfy A1. I will show further that reasonable decision rules in this environment satisfy all of the Simplicity Conditions. 

First, I will discuss how A1 is related to second order stochastic dominance rankings of the error distributions. Let $\tau = \{ T_i\}_{i=1}^N$ be a act $f$'s partition, and let $\tau' = \{T', T_3,\dots, T_N\}$ where $T' = T_1 \cup T_2$. It seems intuitive that learning about acts measurable with respect to $\tau'$ will be easier than learning about those measurable with respect to $\tau$, since the set of events to which the DM must assign probabilities is strictly smaller (in the inclusion order).  Proposition \ref{prop3.2} formalizes the sense in which this is true. Of course, the difficulty of a comparison depends not only on the partitions of the acts involved, but also on their values. Let $\tilde{f}$ be an act such that $\tilde{f}(x) = \e[f|T']$ for $x\in T'$, and $\tilde{f}(x) = f(x)$ otherwise. 

For any act $h$, let $\varepsilon_h(\chi) = \hat{\e}_{\chi}[h] - \e[h]$, and let $G_{h}$ be the distribution of $\varepsilon_h$. Say that distribution $F$ \textit{strictly second-order stochastically dominates} distribution $H$ ($F >_{SOSD} H$) if $\int u(x)dF(x) \geq \int u(x)dH(x)$ for all concave $u$, with strict inequality if $u$ is strictly concave. 
\begin{proposition}\label{prop3.2}
$G_{\tilde{f}}$ strictly second-order stochastically dominates $G_{f}$.
\end{proposition}

The intuition for this result is straightforward. Consider the errors made in the estimation of $\e[f]$ versus $\e[\tilde{f}]$. Notice that for all datasets $\chi$, $\hat{\e}_{\chi}[\tilde{f}|T'] = \e[\tilde{f}|T']$. The randomness of $\hat{\e}_{\chi}[f|T']$ simply adds noise to the distribution of errors from the estimation of $\e[f]$, relative to those of $\e[\tilde{f}]$. A similar conclusion holds for $\bar{f}$, where $\bar{f} = \hat{\e}_{\chi}[f|T']$ on $T'$, and $\bar{f}= f$ elsewhere.

\begin{corollary}\label{lem:hatGsosd}
$\hat{G}_{\bar{f}} >_{SOSD} \hat{G}_f$
\end{corollary}

This follows immediately from Proposition \ref{prop3.2}, since the bootstrap estimator treats the empirical distribution as if it were the true distribution. We return now to the DM's problem of choosing between the act $f$ and a constant act $c$. A natural assumption about decisions in this framework is that greater uncertainty about $\e[f]$ makes it harder for the DM to compare $f$ to constant acts (comparable in this setting meaning that the DM is willing to state a preference for one of the two acts).

\vspace{1.5mm}
\noindent CONFOUNDING SAMPLING UNCERTAINTY: For any acts $f, f'$ with $\hat{\e}_{\chi}[f] = \hat{\e}_{\chi}[f']$ and $\hat{G}_{f'} >_{SOSD} \hat{G}_{f}$ and any constant act $c$, if $f$ is comparable to c then so is $f'$.

\begin{proposition}\label{prop:3.3}
Any decision rule satisfying Confounding Sampling Uncertainty will satisfy A1
\end{proposition}

Proposition \ref{prop:3.3} is an immediate implication of \Cref{lem:hatGsosd}. As a concrete example of a protocol satisfying Confounding Sampling Uncertainty, consider the following. Given data set $\chi$ the DM concludes that $f \succsim c$ if and only if $\e_{\hat{G}_f}\left[\phi\left(\hat{\e}[f]\right)\right] \geq \phi(c) - k$ for some increasing, strictly concave function $\phi$ and constant $k > 0$. Moreover $c \succsim f$ if $c \geq \hat{\e}_{\chi}[f]$. I will refer to this as the \textit{smooth sampling uncertainty} model. It is similar in spirit to the smooth ambiguity model of \cite{klibanoff2005smooth}, where here the higher order uncertainty derives from sampling uncertainty, rather than subjective ambiguity. Additionally, the smooth sampling uncertainty model allows the DM to express incomplete preferences. 

Clearly, $f$ will have a certainty equivalent only if $\hat{G}_f$ is not too dispersed. In particular, $f$ will have a certainty equivalent $c$ if and only if $i$) $\hat{\e}_{\chi}[f] = c$, and $ii$) $\e_{\hat{G}_f}\left[\phi\left(\hat{\e}[f]\right)\right] \geq \phi(c) - k$. \Cref{lem:hatGsosd} implies that the smooth sampling uncertainty model satisfies Confounding Sampling Uncertainty since $\phi$ is concave. In fact, the model satisfies all of the Simplicity Conditions.

\begin{proposition}\label{prop:smoothsampling}
Smooth sampling uncertainty satisfies the Simplicity Conditions. 
\end{proposition}

Smooth sampling uncertainty is by no means the only model for which the Simplicity Conditions will be satisfied. As the proof of Proposition \ref{prop:smoothsampling} shows, A1 and A3 are very close from a learning perspective. Given Proposition \ref{prop:3.3}, A3 will be satisfied so long as the decision rule is suitably continuous in the estimated error distribution $\hat{G}$.

\section{Extension: Completing preferences}\label{sec4}

In many settings the decision maker is forced to make a choice. In such cases we would like to be able to make predictions about behavior even when the environment contains pairs of alternatives that are not ranked according to the incomplete preferences above. Following \cite*{gilboa2010objective}, henceforth GMMS, I assume that the decision maker is characterized by a pair of binary relations ($\succsim, \succsim'$), interpreted as \textit{objective rationality} and \textit{subjective rationality} relations respectively.\footnote{According to GMMS, a choice is objectively rational if the DM can convince others that they, the DM, is right in making it. It is subjectively rational if the DM cannot be convinced that they are wrong in making it.} I consider decision makers with objectively rational preferences that can be represented as in Theorem \ref{thm3.1}, and consider various assumptions on the subjective relation that lead to distinct complete preference relations and representations. This approach allows me to separate the decision makers ability to compare acts from their attitude towards choices between acts that they do not know how to compare. The first point is addressed by Theorem \ref{thm3.1}. Attitudes towards the unknown can be captured by intuitive axioms. Following GMMS, I first make the natural assumption that the subjective relation never reverses the objective. 

\vspace{1.5mm}
\noindent CONSISTENCY: $f \succsim g$ implies $f \succsim' g$.

\vspace{1.5mm}
This notion of consistency allows for indifference according to the subjective relation between acts that are strictly ranked according to the objective relation. A stronger notion of consistency rules out such differences when the acts in question are well-understood. Below I discuss the reason for imposing Strong Consistency only on well-understood acts. 

\vspace{1.5mm}
\noindent STRONG CONSISTENCY FOR SIMPLE ACTS: For any two well-understood acts $f,g$, $f \succsim g \Leftrightarrow f \succsim' g$.

\vspace{1.5mm}
Finally, I assume a cautious approach to incomparable alternatives. 

\vspace{1.5mm}
\noindent CAUTION: For all $f \in F$ and $h \in F_c$, if $f \not \succsim h$ then $h \succsim' f$. 

\vspace{1.5mm}
Caution is exactly the axiom used by GMMS to derive max-min expected utility (MEU) as the subjective relation given Bewley objective preferences. When the objective relation is Simple Bounds, Caution yields a representation in which acts are evaluated according to their simple lower bounds.

\begin{definition}
Preference $\succsim'$ has a \textbf{Cautious} representation if there exists an integer $N$, probability $P$ on $\Sigma$, and a non-constant function $u: Z \rightarrow \mathbb{R}$ such that, for every $f,g \in F$
    \begin{equation*}
       f \succsim' g \ \ \text{iff} \ \ \int_{\Omega} E_{\siminf_{N,f}(\omega)}u \ dP(\omega) \geq \int_{\Omega} E_{\siminf_{N,g}(\omega)}u \ dP(\omega).
    \end{equation*}
\end{definition}

\noindent Call preferences that admit a Cautious representation Cautious preferences.

\begin{theorem}\label{thm4.1}
The following statements are equivalent:
\renewcommand{\theenumi}{\roman{enumi}}
\begin{enumerate}
    \item $\succsim$ satisfies the \textbf{Basic Conditions}, \textbf{Simplicity Conditions}, \textbf{S-Independence}, and \textbf{Uniform Comparability}; $\succsim'$ satisfies \textbf{Archimedean Continuity}; and $\succsim$ and $\succsim'$ jointly satisfy \textbf{Consistency}, \textbf{Strong Consistency for Simple Acts} and \textbf{Caution}.
    \item $\succsim$ has a Simple bounds representation and $\succsim'$ has a Cautious representation, with common parameters $P,u,N$. Moreover, $P$ is unique, $u$ is unique up to positive affine transformations, and for all $f \in F$, $\siminf_{N,f}$ and $\simsup_{N,f}$ are non-empty.
\end{enumerate}
\end{theorem}

Cautious preferences capture a particular attitude towards the unknown. Faced with a difficult choice, the decision maker takes a worst case view of the set of payoffs that they consider reasonable. The properties of Cautious preferences, and their relationship to models of ambiguity aversion, are discussed in detail in \Cref{sec4.2}. In brief, the Cautious completion is the most ambiguity averse completion of Simple Bounds preferences, in the sense of \cite{ghirardato2002ambiguity}. Importantly, Cautious preferences do not satisfy the Uncertainty Aversion axiom of \cite{gilboa1989maxmin}, but do satisfy a modified version, $N$-Ambiguity Aversion, discussed in \cite{hartmann2019hierarchy}. Under some conditions, Cautious preferences are a special case of the Revealed Reasoning model of \cite{saponara2020reasoning}. This connection is discussed in detail in \Cref{sec10}.

\subsection{Properties of Caution}\label{sec4.2}

Cautious preferences can also be understood through the violations of Independence that they exhibit. A decision maker with Cautious preferences may violate independence when the mixture of two acts leads to an act that is harder to approximate from below by simple acts. In particular, the mixture of two $N$-simple acts between which the decision maker is indifferent will in general not be $N$-simple, and may thus be considered inferior to the original acts. Cautious preferences are ambiguity averse in the sense of \cite{ghirardato2002ambiguity}.\footnote{Simple Bounds preferences are also ambiguity averse, if we extend the notion to incomplete preferences.} In particular, for fixed $u$ and $P$, decision makers with a higher capacity are less ambiguity averse than those with lower capacity. 

Given the discussion of \cite{bewley2002knightian} above and the formal parallel between this paper and GMMS, it is natural to consider the relationship between Cautious preferences and \cite{gilboa1989maxmin} MEU preferences. Preferences have an MEU representation if there exists a convex set $C$ of probabilities on $\Sigma$ and a non-constant function $u: Z \rightarrow \mathbb{R}$ such that $f$ is preferred to $g$ if and only if 
\begin{equation*}
    \min_{p \in C} \int_{\Omega} E_{f(\omega)}dp(\omega) \geq \min_{p \in C} \int_{\Omega} E_{g(\omega)}dp(\omega).
\end{equation*}
As GMMS show, such preferences are derived from Bewley preferences through the Consistency and Caution axioms, just as Cautious preferences are derived from Simple Bounds. These representations do not coincide however. This can be seen most easily by considering the Uncertainty Aversion axiom of \cite{gilboa1989maxmin}.

\vspace{1.5mm}
\noindent UNCERTAINTY AVERSION: For every $f,g \in F$, if $f \sim g$ then $(1/2) f + (1/2)g \succsim g$. \vspace{1.5mm}

Uncertainty Aversion captures the notion that mixing acts smooths payoffs, and thus reduces exposure to uncertainty. It is easy to see that this axiom is violated by Cautious preferences. As noted above, mixing any two $N$-simple acts between which the decision maker is indifferent leads to an act which is no better than either of the original acts. This is because the mixed act may not be $N$-simple, and will thus be approximated from below, whereas the original acts were perfectly understood. More generally, whenever the two acts considered are comonotonic, Cautious preferences will satisfy the reverse  Uncertainty Aversion: if $f \sim g$ then $g \succsim (1/2) f + (1/2)g $. Two acts $f,g$ are said to be comonotonic if there is no $\omega, \omega' \in \Omega$ such that $f(\omega) \succ f(\omega')$ and $g(\omega') \succ g(\omega)$. \footnote{This can be seen as follows. For a set $I \subseteq \Omega$ and act $f$ define $\und{f}(I) = \inf\{u\circ f(\omega): \omega \in I\}$. Then the value of $\siminf_{N,f}$ is given by 
\begin{equation*}
\max_{\{I_i\}_{i=1}^N \in T^N(\Omega)} \sum_{i=1}^N P(I_i)\und{f}(I).
\end{equation*}
When $f$ and $g$ are comonotonic $\und{f \alpha g}(I) = \alpha \und{f}(I) + (1-\alpha)\und{g}(I)$ for all $I$, where $f \alpha g = \alpha f + (1-\alpha)g$. The claim follows.} The preference against the mixture will be strict unless there are elements of $\siminf_{N,f}$ and $\siminf_{N,g}$ which share the same partition. The intuition for this reversal is that mixing between comonotonic acts does not smooth payoffs across states in the same way that mixing between non-conmonotonic acts can. On the other hand, since the two acts may have had a lot of variations on different regions of the state space, the mixed act will be more difficult to approximate from below than any of the two acts individually. This contrast illustrates that MEU and Cautious preferences capture very different notions of aversion to uncertainty. 

Uncertainty Aversion is satisfied by Cautious preferences whenever the mixture of $f$ and $g$ is $N$-simple, and indeed the inequality may be strict. A modification of Uncertainty Aversion along these lines is studied by \cite{hartmann2019hierarchy}, who state the following axiom

\vspace{1.5mm}
\noindent N-AMBIGUITY AVERSION: $f_1, \dots, f_n \in F$, $\alpha_1,\dots, \alpha_n \geq 0$, $\sum_{i=1}^n \alpha_i$, $\sum_{i=1}^n \alpha_i f_i = f \in F_N$ such that $f_1 \sim f_2 \dots \sim f_n$ implies $f \succsim f_1$.  

\vspace{1.5mm}
The preferences studied by \cite{hartmann2019hierarchy} differ from Cautious preferences, primarily in that they satisfy the Comonotonic Independence axiom of \cite{schmeidler1989subjective}.  

\vspace{1.5mm}
\noindent COMONOTONIC INDEPENDENCE: For all pairwise comonotonic acts $f,g,h$, and all $\alpha \in (0,1)$, $f \succ g$ implies $\alpha f +(1-\alpha)h \succ \alpha g + (1-\alpha)h$. 

\vspace{1.5mm}
Schmeidler's motivation for Comonotonic Independence, as opposed to the usual Independence axiom, is similar to my motivation of $S$-Independence. \cite{schmeidler1989subjective} notes that arbitrary mixtures may lead to acts which define a much finer (larger) algebra than the original acts, and thus violations of Independence may occur. However, if ``$f$, $g$ and $h$ are pairwise comonotonic, then the comparison of $f$ to $g$ is not very different from the comparison of [$\alpha f +(1-\alpha) h$ to $\alpha g + (1-\alpha)h$]''. In some circumstances however, comonotonicity may be too weak a notion to guarantee that mixtures do not alter the the acts under consideration in ways that lead to violations of independence. Consider the following simple example. Let the state space be the interval $[0,1]$ and identify $u(f(\omega))$ with $f(\omega)$. Let $f = 2$ on $[0,1/2)$ and $f = 4$ on $(1/2, 1]$. Let $g = 1$ on $[0,2/3)$ and $g = 6$ on $(2/3,1]$. Suppose $f \succ g$. The mixture $1/2 f + 1/2 g$ has both a greater spread between its highest and lowest outcomes and a larger partition than $f$. Both these changes may lead the DM to favor $g$. The Simplicity Conditions, along with $S$-Independence, can be seen as making precise the notion of ``not very different''. Cautious preferences capture aversion to uncertainty in the sense that the decision maker, faced with an act that she does not fully understand, assigns to the act the minimum value consistent with monotonicity and her preferences over the acts that she understands well. Uncertainty here stems directly from the complexity of the acts under consideration. However, cautious preferences satisfy the following condition, capturing a notion of aversion to comonotonic mixtures.\footnote{In ongoing work, I characterize a broad class of preferences satisfying Comonotonic Mixture Aversion. I show that within this class, cautious preferences are simultaneously \textit{i}) maximally ambiguity averse, in the sense of \cite{ghirardato2002ambiguity}, \textit{ii}) maximally N-Ambiguity averse, and \textit{iii}) maximally averse to mixtures of $N$-simple acts. Details available upon request.}

\vspace{1.5mm}
\noindent COMONOTONIC MIXTURE AVERSION: For all comonotonic acts $f,g$ with $f \sim g$ and $\alpha \in (0,1)$, $f \succsim \alpha f + (1-\alpha)g $.
\vspace{1.5mm}

\subsection{Empirical evidence}\label{sec:experimental_ambiguity}

Cautious preferences display features of ambiguity aversion. One of the useful features of this model is that it explicitly separates the decision makers understanding of the environment from her attitude towards the unknown. This novel model of ambiguity aversion can explain some experimental findings at odds with existing models. 

\cite{chew2017partial} find that subjects are averse to increases in the number of possible compositions (the size of the state space in my framework) of a deck of cards on which bets are made. For example, consider an individual betting on red in a deck of red and black cards, where the number of red cards is known to be between $n$ and $100 - n$. The authors find that subjects aversion to the ambiguous deck, measured by the difference between their certainty equivalents for an ambiguous act and the corresponding compounded lottery, is decreasing in $n$. The data contradict the predictions of the recursive expected utility model (\cite{klibanoff2005smooth}, \cite{seo2009ambiguity}), which predicts an aversion to increasing the number of possible compositions only in certain cases. Similarly, \cite{viscusi1992bayesian} document increasing ambiguity aversion as the range of uncertain outcomes increases.

\subsection{Alternative approaches}
There may be situations in which decision makers take the opposite approach to what they recognize as their limited understanding of the objects of choices.

\vspace{1.5mm}
\noindent ABANDON: For all $f \in F$ and $h \in F_c$, if $h \not \succsim f$ then $f \succsim' h$.

\vspace{1.5mm}
Replacing Caution with Abandon yields the obvious converse representation, which I call Reckless preferences, whereby acts are evaluated according to their $\simsup$ rather than $\siminf$. Such preferences seem to contradict common assumptions about risk and ambiguity aversion. They have a flavor of, but are distinct from, optimism or overconfidence. The decision maker does not behave exactly as if she thought high payoff states are more likely to occur, but rather as if the underlying act is as good as possible without violating her limited understanding of the situation. Such a model may be useful for understanding the behavior of decision makers who seem to favor nebulous prospects with high potential over those that are well-understood. I will sometimes use the notation $U(f)$ or $U(f,P)$ to denote perceived utility. If the agent is cautious $U(f,P) = E_P[E_{\siminf_{N,f}(\omega)}u]$, and if the agent is reckless $U(f,P) = E_P[E_{\simsup_{N,f}(\omega)}u]$.

An alternative approach to incompleteness would be to proceed as in \cite{bewley2002knightian} and assume that there is a default option and the DM satisfies an inertia condition, whereby an alternative act is chosen over the default if and only if it is preferred, according to the underlying incomplete preference. While a full coverage of inertia is beyond the scope of this paper, the results presented here are helpful for understanding such a model. What matters for choice under inertia is the $\simsup$ of the default act and the $\siminf$ of the alternative. The behavior of these objects is studied in the comparative statics and applications sections.\footnote{The earlier working paper version of this paper includes further discussion of inertia.}

\subsection{A note on Consistency}

Given the interpretation of the subjective preference as an extension of the incomplete objective preference, it would seem natural to impose Strong Consistency everywhere: for all $f,g \in F$, $f \succsim g$ implies $f \succsim' g$, and $f \succ g$ implies $f \succ' g$. Consistency allows the decision maker to be indifferent under $\succsim'$ between some acts that are strictly ranked under $\succsim$.\footnote{In fact, under Strong Consistency for Simple Acts, this occurs for acts $f,g$ only if a) either $f \geq g$ or $g \geq f$, and b) $\siminf_{N,f} \cap \siminf_{N,g} \neq \varnothing$.} It turns out that Strong Consistency everywhere is incompatible with Archimendean Continuity of $\succsim'$. This conflict between Caution, Strong Consistency, and Continuity is not unique to this setting. GMMS face the same trade-off obtaining MEU preferences as the Cautious completion of Bewley preference (see Section \ref{sec4.2}), and the incompatibility holds for a broad class of incomplete preference models.\footnote{As a technical note, Strong Consistency for Simple Acts renders imposition of further basic conditions on $\succsim'$ redundant, as these are inherited from $\succsim$ on the set of well-understood acts.} 
Rather than relaxing consistency to retain continuity, as in Theorem \ref{thm4.1}, one can drop the continuity requirement for $\succsim'$ and impose Strong Continuity everywhere. This yields ``lexicographic'' preferences: acts are first ranked by statewise dominance, and then according to their $\siminf$ if they are not ranked according to statewise dominance. The difference between the lexicographic model and that of Theorem \ref{thm4.1} is one of continuity; the two models have nearly identical properties and predictions.

\section{Applications}\label{sec:applications}

The subsequent applications focus on Cautious and Reckless preferences. However the results are relevant for Simple Bounds preferences more generally, as these are also characterized by the $\siminf$ and $\simsup$. For example in the insurance choice application of \Cref{sec8} under inertia, an individual chooses some plan $f$ over the default act $f$ if and only if $\siminf_{N,f}$ is preferred to $\simsup_{N,g}$. Thus the results on insurance valuation for a Cautious individual in \Cref{sec8} also answer the question of how an individual with inertia will value new offerings relative to their default plan. 

\subsection{Consumption-Savings}\label{sec9}

In this section I study consumption-savings decisions in a two period model. The following observation is useful for understanding the role of complexity in this setting. Let $f$ be an act on $\Omega = \mathbb{R}$. Let $P$ be the cdf of the agent's belief and $u$ the DM's utility function. Fix a selection $s$ from $\siminf_{N,f}$, and recall that since $f$ is monotone, $s$ is defined by a vector of cut-offs $(t_i )_{i=0}^N$. Define the \textit{lower-perceived distribution} $P^s_N$ as the distribution satisfying $E_{P^s_N}[f] = E_P[\siminf_{N,f}]$. Similarly for $s \in \simsup_{N,f}$ defined by cut-offs $(t_i )_{i=0}^N$, the \textit{upper-perceived distribution} $\tilde{P}^s_N$ is defined as the distribution satisfying $E_{\tilde{P}^s_N}[f] = E_P[\simsup_{N,f}]$.\footnote{So if $P$ has no mass points, $P^s_N$ places mass of $P(t_{i+1}) - P(t_{i})$ on $t_i$, for $i = 0,\dots,N-1$ and is zero elsewhere, and $\tilde{P}^s_N$ places mass of $P(t'_{i}) - P(t'_{i-1})$ on $i = 1,\dots,N$. In other words, \Cref{obs9.1} is simply the observation that for an increasing act, the $\siminf$ ($\simsup$) is a step function, where the steps touch $u\circ f$ at their leftmost (rightmost) endpoints.}

\begin{observation}\label{obs9.1}
Let $f\circ u$ be increasing. Then $\tilde{P}^{s'}_N \succsim_{FOSD} P \succsim_{FOSD} P_N$ for any $s \in \siminf_{N,f}$ and $s' \in \simsup_{N,f}$. If, moreover, $u\circ f$ is integrable then $E_P[u\circ f] = \lim_{M\rightarrow \infty} E_{P_M}[u \circ f] = E_{\tilde{P}_M}[u \circ f]$.
\end{observation}

Observation \ref{obs9.1} can be extended to more general environments using the Lebesgue approach of section \ref{sec3.1}. This observation facilitates comparisons between constrained and unconstrained agents.

\subsubsection{Portfolio choice with a safe and risky asset} 

Suppose the DM is cautious and faces the problem of allocating wealth between consumption, purchase of a risk-free asset and purchase of a risky asset. I compare choices under a capacity constraint to those of a fully rational agent ($N = \infty$). Let $R_b$ be the certain return on the risk free asset, and $R_s$ the uncertain return on the risky asset with cdf $P$ and pdf $p$ on the interval $[\und{R}_s, \bar{R}_s]$. The DM with capacity $N$ solves
\begin{equation}
    \max_{b,s} u(w-b-s) + \beta V^N(R_b b + R_s s)
\end{equation}
where
\begin{equation}\label{eq32}
    V^N(R_b b + R_s s) = \max_{\{\{t_i\}_{i=0}^N \in I\}} \ \sum_{i=1}^{N}[P(t_i) - P(t_{i-1})]u(R_bb + t_{i-1}s),
\end{equation}
and $I$ is the set of interval partitions of $[\und{R}_s, \bar{R}_s]$. Note that $\lim_{N \rightarrow \infty} V^N(R_b b + R_s s) = \int u(R_bb + rs) dP(r)$.
Cautious versus unconstrained agents. As Observation \ref{obs9.1} makes clear, for any \textit{fixed} portfolio choice the unconstrained DM perceives a FOSD shift of the returns perceived by the constrained DM. Of course, as the constrained DM changes their portfolio allocation the perceived distribution of $R_s$ induced by $\siminf$ will also change, so this is not the same problem as comparing portfolio choice under two different fixed yield distributions.

\begin{proposition}\label{prop9.1}
Let $u$ be CRRA with a coefficient of relative risk aversion greater than or equal to 1. \footnote{Part (\textit{i}) holds under the weaker condition that $u'''(z)u'(z)/u''(z)^2 \geq 2$ for all $z$. I state it for CRRA for simplicity. The assumption that the risk aversion coefficient is greater than 1 is standard. Part (\textit{ii}) also holds under weaker conditions, CRRA just makes it easy to show that (\ref{eq9.9}) in the proof holds.} Then
\renewcommand{\theenumi}{\roman{enumi}}
\begin{enumerate}
    \item For a given level of aggregate savings $x > 0$, a Cautious DM allocates a greater portion of savings to the safe asset than an unconstrained DM. 
    \item A Cautious DM saves more overall then an unconstrained DM. 
\end{enumerate}
\end{proposition}

As the proof illustrates, the assumption of CRRA with coefficient greater than 1, or the weaker condition in the footnote, is by no means necessary.

\subsubsection{Equilibrium Asset Prices}

Consider a representative agent model in which the agent chooses between a safe and risky asset, which must both be in zero net supply in equilibrium. Returns are as in the previous section. Normalize the price of the safe asset to 1 and let the price of the risky asset be given by $p$. I assume that the agent has an endowment of $w$ in each period. Let $P^{b,s}$ be the perceived distribution given investment $b$ in the safe asset and $s$ in the risky asset. The constrained DM solves
\begin{equation*}
    \max_{b,s} u(w - b -ps) + \beta \int u(w + bR_b + sr)dP^{b,s}(r).
\end{equation*}
For the unconstrained DM the problem is identical, except that $P^{b,s}$ is replaced with $P$. 
It is intuitive that a cautious (reckless) constrained DM should be biased towards the safe (risky) asset, relative to the unconstrained DM. The constraint simply coarsens the DM's understanding of the stochastic payoff of the risky asset, which leads the cautious agent to undervalue it and reckless agent to overvalue. In fact, when considering assets that must be in zero net supply in equilibrium, we can extend this intuition to make comparisons between intermediate capacity levels. 
\begin{proposition}\label{prop9.3}
The equilibrium risky asset price with a Cautious (Reckless) representative agent is increasing (decreasing) in the agent's capacity. 
\end{proposition}
When the agent is cautious, Proposition \ref{prop9.3} gives an explanation for the so called ``equity premium puzzle''.

\subsection{Insurance valuation}\label{sec8}

I now compare the salience of various features of insurance plans for a complexity-constrained versus a fully rational ($N = \infty$) individual. I focus on a Cautious individual, but the analysis is equally relevant for an individual who faces inertia with respect to some fixed default plan, as discussed above.

I take as given a class of basic insurance contracts that are characterized by a premium $p$, i.e. price for the plan; a deductible $d$, below which the individual bears all losses; a coverage rate $c$ specifying the fraction of losses above the deductible covered by the plan; an out-of-pocket expenditure cap $m$, which is the maximum amount that an individual will have to pay, excluding the premium. I do not provide a foundation for the use of these piece-wise linear contracts, but they are by far the most common form of insurance contract. These contracts have the feature that the individual's ex-post wealth is decreasing in the realized loss in $[0,\Bar{\omega}]$. Recall that in this case $\siminf$ and $\simsup$ will be measurable with respect to $N$-element interval partitions, which I will describe by the cut-off states $t_0 = 0, t_1 \dots, t_N = \Bar{\omega}$. 

\subsubsection{The setting}

Consider an individual facing a bounded loss distribution on $\Omega = [0,\Bar{\omega}]$ with absolutely continuous CDF $P$. The loss distribution will be fixed throughout. Let $w$ be the individual's endowment wealth. In autarky the ex-post wealth is $w-\omega$, where $\omega$ is the realized loss. 

I do not consider here the choice of the contract by the insurer, only the valuation by the individual. Say that a capacity-constrained individual \textit{over-reacts} (\textit{under-reacts}) to a change from one plan to another if the magnitude of the difference between their values for the two plans is greater (less) than that of a fully rational individual.

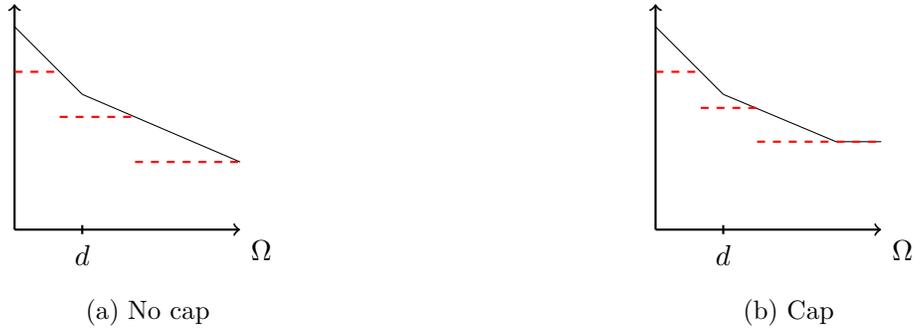
\begin{figure}[ht]
    \centering
    \begin{subfigure}{.45\textwidth}
    \centering
    \begin{tikzpicture}[scale = .6]
        \draw[thick,->] (0,0) -- (0,5);
        \draw[thick,->] (0,0) -- (5,0) node[anchor = north west] {$\Omega$};
        \draw[thick] (1.5,.1) -- (1.5,-.1) node[anchor = north] {$d$};
        \draw (0,4.5) -- (1.5,3);
        \draw (1.5,3) -- (5,1.5);
        \draw[red,thick,dashed] (0,3.5) -- (1,3.5);
        \draw[red,thick,dashed] (1,2.5) -- (2.67,2.5);
        \draw[red,thick,dashed] (2.67,1.5) -- (5,1.5);
    \end{tikzpicture}
    \caption{No cap}
    \label{fig:subim1}
    \end{subfigure}
    \hfill
    \begin{subfigure}{0.45\textwidth}
    \centering
    \begin{tikzpicture}[scale = .6]
        \draw[thick,->] (0,0) -- (0,5);
        \draw[thick,->] (0,0) -- (5,0) node[anchor = north west] {$\Omega$};
        \draw[thick] (1.5,.1) -- (1.5,-.1) node[anchor = north] {$d$};
        \draw (0,4.5) -- (1.5,3);
        \draw (1.5,3) -- (4,1.95);
        \draw (4,1.95) -- (5, 1.95);
        \draw[red,thick,dashed] (0,3.5) -- (1,3.5);
        \draw[red,thick,dashed] (1,2.7) -- (2.25,2.7);
        \draw[red,thick,dashed] (2.25,1.95) -- (5,1.95);
    \end{tikzpicture}
    \caption{Cap}
    \label{fig:subim2}
    \end{subfigure}
    \caption{$\siminf$ with and without out-of-pocket maximum}
\label{fig:image2}
\end{figure}

\subsubsection{Salient plan features}

I will consider changes in the deductible and in the coverage rate, meaning the percentage $c$ of losses above the deductible covered by the contract ($c=1$ in the full insurance case). Proposition \ref{prop8.1} states that a cautious agent over-reacts to changes in both parameters when there is no out-of-pocket expenditure cap. On the other hand, when there is a binding cap the agent's valuation is unaffected by changes in the coverage rate, provided the coverage rate is sufficiently high. 

\begin{proposition}\label{prop8.1}
Consider marginal changes to a contract with no out-of-pocket expenditure cap. A \textbf{cautious} agent:
\begin{enumerate}
    \item over-reacts to changes in the deductible.
    \item over-reacts to changes in the coverage rate. 
\end{enumerate}
\end{proposition}

Proposition \ref{prop8.1} is silent on which of the distortions is relatively larger. If we focus on marginal changes to a baseline plan with full insurance above the deductible, then the over-reaction to the coverage rate will be relatively larger than that to the deductible if the decision maker believes that high losses are sufficiently likely. In this case the highest cut-off defining the $\siminf$ will be close to $d$ (but always strictly below $d$, by  Lemma \ref{lem8.1}), and so the response to a deductible change will be close to that of a fully rational DM. However caution will still cause the DM to drastically over-react to reductions in coverage rate. Similarly, when $N$ is sufficiently high the highest cut-off defining the $\siminf$ will be close to $d$, so again the over-reaction to the coverage rate will be the larger of the two when $N$ is high enough.

Restricting attention to contracts with near full insurance above the deductible, it is easy to show a stronger comparative statics result.

\begin{proposition}\label{prop8.2}
Consider a contract with deductible $d$ and $c=1$. For both the deductible and the coverage rate the magnitude of the response of a cautious agent to marginal changes is decreasing in $N$.
\end{proposition}

Many of the settings in which biases towards low deductibles are observed involve full coverage above the deductible. In this case cautious individuals are always biased towards low deductible plans, and this bias is decreasing in their capacity. 

\begin{proposition} \label{prop8.6}
With $c=1$, the amount a cautious individual is willing to pay to lower the deductible by a given amount is decreasing in their capacity. 
\end{proposition}

Finally, I consider how a capacity constraint affects the individual's willingness to pay to decrease the out of pocket maximum. 

\begin{proposition}\label{prop8.5}
The amount a cautious agent would be willing to pay to decrease the out-of-pocket maximum is decreasing in their capacity. 
\end{proposition}

The finding that lower capacity individuals place a greater value on the out of pocket maximum may help explain the widely documented bias towards full insurance, for example by \cite{shapira2008preference}. Full insurance plans are easy to understand, since the loss is independent of the state. In general, constrained individuals will overvalue full insurance plans relative to those for which the realized losses are a more complicated function of the state.

\subsubsection{Dominated choices}

\cite{bhargava2017choose} find that many individuals choose dominated plans, and that the propensity to do so is positively correlated with both high expected losses and earnings levels.\footnote{Interestingly, the authors find that the number of parameters needed to describe a plan, an alternative measure of complexity encountered in the literature, does not predict dominated choices. This suggest that the partitional notion of complexity may be more relevant in this setting.} The cautious model can provide an explanation for these observations.\footnote{ My model does not predict strictly dominated choices, but allows for indifference between pairs of plans ordered by weak dominance, even for full support beliefs. If dominated choices are related to act complexity then the empirical evidence is informative about what types of acts are perceived to be complex. Additional factors, such as difficulty understanding how plans map states to payments, may combine with complexity considerations to produce strictly dominated choices.}  The result depends on the nature of plan dominance. These results also demonstrate the usefulness of the comparative statics properties discussed in the appendix. \cite{bhargava2017choose} observe individuals  choosing low deductible plans even when the increase in the premium relative to a high deductible plan (holding other plan features constant) is greater than the maximum possible savings from the lower deductible. Figure \ref{fig2} illustrates a situation in which the maximum possible savings from the low deductible, high premium plan is equal to the increase in the premium. Call this as a \textit{weakly dominated low-deductible plan}.

\begin{figure}[h]
 \centering
    \begin{tikzpicture}[scale = .7]
        \draw[thick,->] (0,0) -- (0,3.75);
        \draw[thick,->] (0,0) -- (5,0) node[anchor = north west] {$\Omega$};
        \draw[thick] (1,.1) -- (1,-.1) node[anchor = north] {$d$};
        \draw[thick] (1.75,.1) -- (1.75,-.1) node[anchor = north] {$d'$};
        \draw (0,3.5) -- (1.75,1.75);
        \draw (1.75,1.75) -- (5,0.5);
        \draw[very thick, dotted] (0,3.05) -- (1,2.05);
        \draw[very thick, dotted] (1, 2.05) -- (1.75, 1.75);
        \draw[red,thick,dashed] (0,1.65) -- (2,1.65);
        \draw[red,thick,dashed] (2,1.05) -- (3.5,1.05);
        \draw[red,thick,dashed] (3.5,0.5) -- (5,0.5);
    \end{tikzpicture}
    \caption{Indifference with a weakly dominated low-deductible plan}
    \label{fig2}
\end{figure}
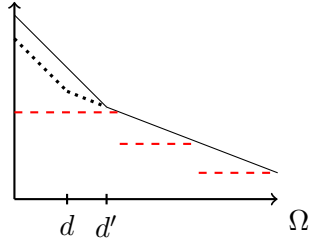

In Figure \ref{fig2} the solid black line is the high deductible, low premium plan, and the black dashed line is a low premium, high deductible plan. Both plans have the same coverage rate. The red dotted line is the $\siminf$, which in this case is the same for both plans. This occurs whenever the lowest cut-off for the high deductible plan is above its deductible. Let $l(\omega|d,c)$ be the amount paid by the consumer when the loss is $\omega$ given a contract with deductible $d$ and coverage rate $c$ (this is stated formally in Lemma \ref{lem8.4} in the Online Appendix). More interesting than the fact that dominated plans can be chosen is are the conditions which are conducive to such mistakes. \cite{bhargava2017choose} observe that dominated choices are correlated with both expected losses and earnings levels. Proposition \ref{cor8.2} predicts the former. To the extent that earnings are correlated with the capacity to evaluate acts, Proposition \ref{cor8.3} predicts that low earners will be more prone to mistaken indifference. The following corollary establishes a single-crossing property of dominated choice which implies that individuals who are more pessimistic about their losses are more prone to make dominated choices (in both cases, all other model parameters are held constant).

\begin{proposition}\label{cor8.2}
If an individual with belief $P$ chooses a weakly dominated low-deductible plan then so does one with belief $P'$ if $P' \succeq_{MLR} P$. 
\end{proposition}
\begin{proposition}\label{cor8.3}
If an individual with capacity $N'$ chooses a weakly dominated low-deductible plan then so is one with capacity $N$ if $N' \geq N$. 
\end{proposition}

\subsubsection{Discussion}

This section relates to the large empirical literature documenting behavioral phenomena in choices of complicated contracts. The model is able to explain many observed choice patterns which differ from the predictions of standard EU theory. \cite{abaluck2011choice} find that consumers underweight out-of-pocket spending relative to premiums. Moreover, cautious agents will respond more than the fully rational to changes in the coverage rate when there is an out-of-pocket expenditure cap and the coverage rate is high (see Proposition \ref{prop8.1}). \cite{cutler2004extending} document a bias towards low deductible plans. This is consistent with a cautious decision maker (see Proposition \ref{prop8.1}). The results of this section highlight the dependence of the qualitative nature of comparisons to the rational model on details of the insurance plans under consideration. Such variation forms an interesting basis for further empirical work, and provides a more nuanced perspective on ``behavioral biases''.\footnote{In Online Appendix \ref{sec:MEUinsurance} I discuss how the predictions of the Cautious model in this setting relate to those of MEU.}

\subsection{Contracting under moral hazard}\label{sec:principal_agent} 

An agent exerts effort $a$ which induces a distribution of output $P_a$, with support on a bounded set $\Omega$. The principal observes output but not effort, and so offers the agent a wage schedule $w : \Omega \rightarrow \mathbb{R}$. The agent's ex-post payoff is given by $u(w,a)$. The agent's outside option payoff is normalized to $0$.

The principal chooses $w$ to maximize $\e_{P_{a^*(w)}}[v(\omega, w(\omega))]$, where $a^*(w)$ is the agent's effort choice given contract $w^*$ (assume the agent chooses the principal's preferred effort level when the agent's problem has multiple solutions). Assume that $w \mapsto u(w,a)$ is continuous and increasing for all $a$, and that $ w \mapsto v(\omega,w)$ is decreasing for all $\omega$.\footnote{The results extend immediately to the case of multi-dimensional payment spaces $W$, provided the monotonicity conditions on $u$ and $v$ hold with respect to the same order on $W$ for all $a$, and that the order does not depend on $a$.}

\subsubsection{Cautious Agents}

Let the agent be Cautious, with capacity $N$. Suppose that the principal offers a contract $w$ and the agent chooses effort $a$. I will refer to $( w,a )$ as the act induced by $w$ and $a$ (where the only uncertainty is on the $w$ coordinate). If the principal offers the agent a contract $\hat{w}$ such that $(\hat{w},a) \in \siminf_{N,( w,a ),P_a}$ then agent is no worse off; the agent can always choose the same effort level and is indifferent between the acts $( w,a )$ and $(\hat{w},a ) \in \siminf_{N,( w,a ), P_a}$ given output distribution $P_a$. In fact the nature of Cautious preferences implies that the agent's optimal effort choice will not change. 

\begin{proposition}\label{prop6.1}
Given a contract $w$ that induces optimal effort choice $a$ by the agent, any contract $\hat{w}$ such that $( \hat{w},a ) \in \siminf_{N,( w,a ),P_a}$ induces the same effort level and has the same value for the agent. 
\end{proposition}

\begin{proof}
Suppose the agent is offered contract $w$ and chooses effort $a$. Assume $\siminf_{N,( w,a ),P_a}$ is unique (the same argument applies to any selection). If the principal instead offers $\Hat{w}$ such that $( \hat{w},a ) = \siminf_{N,w,P_a}$ and effort remains unchanged then the agent's perceived payoff is unchanged, $U(( w,a ),P_a) = U(( \hat{w},a ), P_a)$.

Since $w \mapsto u(w,a)$ is increasing, it must be that $w(\omega) \geq \hat{w}(\omega)$ for all $\omega$. So for any $a'$, $( \hat{w}, a')$ is an $N$-simple act that satisfies $u(w(\omega),a') \geq u(\hat{w}(\omega),a')$ for all $\omega$. In other words, for all $a'$, $( \hat{w}, a')$ is in the set of acts for which $\siminf_{N,( w,a' ), P_{a'}}$ is maximal. Then by the definition of $\siminf$, we have $U(( \hat{w}, a' ), P_{a'}) \leq U(\siminf_{N,( w,a'),P_{a'}}, P_{a'}) = U(( w,a'), P_{a'})$ for all $a'$. Since the agent originally chose $a$ it must be that $U(( w,a') ,P_{a'}) \leq U(( w,a),P_a)$. Combining the inequalities in the two preceding lines, we have $U(( \hat{w}, a' ), P_{a'}) \leq U(( w, a ), P_{a}) = U(( \hat{w},a), P_a)$, where the final equality follows from the definition of $\hat{w}$. This proves that effort $a$ remains optimal.
\end{proof}

Since $w$ statewise dominates $\hat{w}$ and $w \mapsto v(w,\omega)$ is decreasing we have the following immediate corollary of Proposition \ref{prop6.1}.  
\begin{corollary}\label{cor6.1}
If the agent has capacity $N$ and is cautious then all principal-optimal contacts are $N$-simple.
\end{corollary}

\subsubsection{Reckless Agents}

Unsurprisingly, when the agent is reckless the principal is able to exploit the agent by confusing them with a complex contract.

\begin{lemma}\label{lem:reckless_complex}
For any $N$-simple contract $f$, there exists an $N$-complex $g$ such that the Reckless agent is indifferent between $f$ and $g$, and exerts the same effort. 
\end{lemma}

More importantly, we can say more about the types of contracts that can be used to exploit a Reckless agent: the principal will offer ``prizes'', discontinuous jumps in the wage schedule at high-payoff states. The agent sees these prizes and over-reacts, overvaluing the resulting payoffs. 

To illustrate this feature, assume that $\Omega = [0,1]$ and $P_a$ is absolutely continuous and has full support for all $a$. For simplicity, restrict attention to increasing contracts. Say that $f$ is an $\varepsilon$-\textit{bait contract} if it is discontinuous in an $\varepsilon$ neighborhood of $\omega = 1$. A simple contract is one with a finite range. 

\begin{proposition}\label{prop6.3}
For any increasing contract $f$ that is either simple or continuous, and any $\varepsilon > 0$, there exists an $\varepsilon$-bait contract $g$ such that the Reckless agent is indifferent between $f$ and $g$, and exerts the same effort.
\end{proposition} 

This phenomenon is similar in spirit to that studied by \cite{viero2014bait}, in that the principal exploits the limited understanding and ``optimism'' of the agent by offering bait contracts.

\section{Related Literature}\label{sec10}

Other papers have built on the idea that acts with many outcomes may be difficult for a decision maker to evaluate. \cite{neilson1992some} proposes a model of choice under risk in which the decision maker uses a different utility function when computing expectations for lotteries with different support sizes. 

\cite{puri2020simplicity} axiomatizes a ``Simplicity Representation'' of choice under risk, in which a lottery $p$ is evaluated according to $E_p[u(x)] - C(|support(p)|)$, for some increasing function $C$.\footnote{Related ideas appear in the menu choice literature. \cite{ortoleva2013price} axiomatizes a model of preferences over lotteries of menus in which, similar to the model of \cite{puri2020simplicity}, the decision maker attaches a cost to lotteries with more menus in their support.  } While this model also relates support size to complexity, its empirical content is quite different; the Simplicity Representation and Simple bounds are very far from being ``dual'' in the sense one might expect at first glance. Most importantly, the Simplicity Representation makes a sharp separation between the values on which a lottery is supported and its complexity, as measured by support size. As a consequence, the Simplicity Representation predicts potentially extreme preference reversals resulting from small changes; arbitrarily small perturbations of a lottery that increase the size of its support can dramatically change its complexity cost. This is not the case in the Simple Bounds model. Unsurprisingly, the Simplicity Representation also predicts violations of first order stochastic dominance, analogous to violations of Monotonicity in the current setting. 

\cite{gul2014expected} study a model in which there is a $\sigma$-algebra of ``ideal'' events $\mathcal{E}$, which can be thought of as well-understood, and acts are bracketed by their upper and lower bounds among those measurable with respect to $\mathcal{E}$. In contrast, A0 implies that in the Simple Bounds model, all events are well-understood. Complexity here is about aggregation of payoffs across different events, rather than the events themselves. There is no $\sigma$-algebra $\mathcal{E}$ for which the set of $N$-simple acts is the set of $\mathcal{E}$-measurable acts (aside from the trivial case of $N=\infty$). Moreover, \cite{gul2014expected} consider an aggregation of the bounding acts which corresponds to complete preferences. 

\cite{saponara2020reasoning} axiomatizes a Revealed Reasoning model which is close in spirit to the Cautious model. In this model a decision maker is characterized by a set of partitions $\mathcal{P}$, and evaluates an act $f$ according to the best act that is uniformly below $f$ and measurable with respect to some partition $P \in \mathcal{P}$. $\mathcal{P}$ need not be characterized by a fixed number of elements; instead the axioms impose that $\mathcal{P}$ satisfy a richness condition.\footnote{In some cases, this richness condition is in fact incompatible with $\mathcal{P}$ being equal to the set of $N$-element partitions, for some $N$.} Moreover, preferences over the set of acts that are measurable with respect to some $P\in \mathcal{P}$ may not have an expected utility representation. Thus the Revealed Reasoning model differs from Caution in important ways, the latter providing more choice-based structure on the DM's behavior. On a more technical note, \cite{saponara2020reasoning} implicitly begins with the assumption that an act's complexity is determined by its partition. One of the contributions of the current paper is to derive this conclusion from assumptions on choice behavior (primarily in Proposition \ref{prop3.4}). These differences aside, the relationship between the Simple Bounds and Revealed Reasoning models is similar to that between Bewley preferences and MEU.

\cite{ahn2010framing} also study preferences in which partitions play a central role. In their partition-dependent expected utility (PDEU) representation the decision maker uses a different belief to evaluate acts depending on the partition used to describe the state space. This can lead to preference reversals between acts $f$ and $g$ when different partitions (with respect to which $f$ and $g$ are measurable) are used to describe the state space. 

There are formal similarities between this paper and \cite*{gilboa2010objective} and \cite{lehrer2014extension}. As in \cite{lehrer2014extension}, I begin by characterizing preferences on a small set of acts and extend these to a larger subset, although not necessarily to all acts. As in GMMS, given a characterization of incomplete preferences, I show that under additional assumptions a complete preferences relation consistent with it exists, and takes a specific form. In GMMS the incomplete and complete preference relations admit representations \`{a} la \cite{bewley2002knightian} and \`{a} la \cite{gilboa1989maxmin} respectively. The incomplete preferences arising from the partition size notion of complexity do not admit a Bewley representation. Similarly the Cautious completion is not an MEU preference, although it still captures a notion of ambiguity aversion. 

One objective of the current paper is to explore the connection between Ellsberg type phenomenon (\cite{ellsberg1961risk}) arising in the presence of ambiguity and the complexity of decision making problems under uncertainty. As in \cite{segal1987ellsberg} and \cite{klibanoff2005smooth}, bets on ambiguous urns are viewed as a two stage act, where the first stage is subject to uncertainty. That ambiguity may arise from complexity in such an environment is not a new idea. \cite{gilboa1989maxmin} states ``One conceivable explanation of this phenomenon [Ellsberg-type preferences] which we adopt here is as follows: \dots the subject has too little information to form a prior. Hence (s)he considers a set of priors as possible.'' My paper is novel however in that it derives ambiguity averse preferences by explicitly characterizing subjective complexity. \cite{bewley2002knightian} and \cite{gilboa1989maxmin} relax completeness and independence respectively. I do both, but in a way that is driven by explicit assumptions about perceived complexity. 

My applications relate to a number of papers studying bounded rationality and ambiguity aversion. There is a large literature on the simplicity of contracts observed in reality. \cite{mukerji1998ambiguity} uses ambiguity aversion to explain contract incompleteness. \cite{anderlini1994incomplete} use a similar notion of contract complexity to my partition size definition. However they essentially impose that contracts must be simple in this sense, where as I show that such contracts are optimal from the principal's perspective when facing a certain type of agent.

\appendix
\section{Appendix: Omitted Proofs}\label{app:omittedproofs}

\subsection{Proposition \ref{prop3.4}} 

\begin{proof} 
Suppose an act $f$, with partition $\{T_i \}_{i=1}^N$, has a certainty equivalent. I wish to show that any other act $g$ with the same partition also has a certainty equivalent. Throughout the proof, label the partition so that $g(T_{i+1}) \succ g(T_{i})$ for all $i$. For notational simplicity, I will identify each act $f$ with it's utility image $v\circ f$. Finally, assume that there exists a constant act $\bar{c}, \und{c}$ such that $\bar{c} \succ g(\omega) \succ \und{c}$ for all $\omega$ (in the end we will establish existence of a certainty equivalent when such a $\bar{c},\und{c}$ do not exist). The proof will proceed by induction on $N$. The induction hypothesis for each $K < N$ is that all acts measurable with respect to a $K$-element coarsening of $\{T_i \}_{i=1}^N$ have certainty equivalents. Note that, by the Basic Conditions, preferences on $F_c$ have an expected utility representation.

Before proceeding to the induction proof, I show the following claim.

\noindent\textit{Claim 1}: if $\lambda g + (1-\lambda)c_1 \sim c_2$ for $c_1,c_2 \in F_c$ with $\bar{c} \succsim c_1 \succsim \und{c}$ and $\lambda \in (0,1)$, then $g$ has a certainty equivalent. The proof of Claim 1 is as follows. By Weak C-Independence, it suffices to show that there exists $c_3\in F_c$ such that $\lambda c_3 + (1-\lambda)c_1 \sim c_2$. If $c_1 \sim c_2$ then we are done. Suppose $c_1 \succ c_2$. By Monotonicity and $c_1 \succsim \und{c}$, $c_2 \succ \und{c}$. Then there exists $\alpha \in (0,1)$ such that $\alpha \und{c} + (1-\alpha)c_1 \sim c_2$. Let $c_3 = \frac{\alpha}{\lambda}\und{c} + \frac{\lambda - \alpha}{\lambda}c_1$. Since $g(\omega) \succ \und{c}$, and given the expected utility representation on $F_c$, Monotonicity implies that $\lambda > \alpha$, so $c_3$ is well defined. If $c_2 \succ c_1$ replace $\und{c}$ with $\bar{c}$. 

Now for the induction proof. I first show that all binary acts have certainty equivalents. For any event $E$, let $f$ be a bet on $E$ that has a certainty equivalent (which exists by A0), and $g$ be another arbitrary bet on $E$. Let $E^c = \Omega \setminus E$. There are a few cases to consider. Suppose $f(E) \succ g(E) \succ g(E^c) \succ f(E^c)$. Then, using the expected utility representation on $F_c$, $ \exists \ \lambda \in (0,1)$  and $\hat{c} \in F_c$ such that $ \lambda f(E^c) + (1-\lambda)\hat{c} \sim g(E^c)$ and $\lambda f(E) + (1-\lambda)\hat{c} \sim g(E)$ ($\lambda = (u(g(E)) - u(g(E^c)))/(u(f(E)) - u(f(E^c)))$). By Weak C-Independence $\lambda f + (1-\lambda) \hat{c}$ is well-understood, and since this act is payoff equivalent to $g$, $g$ is as well. Suppose instead that $g(E) \succ f(E) \succ f(E^c) \succ g(E^c)$. Then, as before, there exist $\lambda,\hat{c}$ such that $\lambda g(E) + (1-\lambda)\hat{c} \sim f(E)$ and $\lambda g(E^c) + (1-\lambda) \hat{c} \sim f(E^c)$, so $\lambda g + (1-\lambda)\hat{c}$ has a certainty equivalent. Then $g$ has a certainty equivalent as well, by Weak C-Independence.

Suppose $g(E) \succsim f(E) \succsim g(E^c) \succsim f(E^c)$. Then $\exists \ \lambda_1, \lambda_2 \in  [0,1)$ such that $\lambda_1 \bar{c} + (1-\lambda_1) f(E) \sim g(E)$, and $\lambda_2 \bar{c} + (1-\lambda_2) f(E^c) \sim g(E^c)$. Let $\lambda = \min\{\lambda_1, \lambda_2 \}$, and suppose WLOG that this is equal to $\lambda_1$. By Weak C-Independence, $f' := \lambda \bar{c} + (1-\lambda)f$ has a certainty equivalent. Note $f'(E) \sim g(E)$. By continuity $\exists \ \alpha \in [0,1]$ such that $\alpha f'(E) + (1-\alpha) f'(E^c) \sim g(E^c)$. Then $\alpha f'(E) + (1-\alpha) f'$ is payoff equivalent to $g$ and has a certainty equivalent by Weak C-Independence, so $g$ has a certainty equivalent. The remaining cases are analogous. Since by A0 every event has a well understood bet, we can conclude that any binary act has a certainty equivalent. This is the first step in the induction. 

Now suppose that $f$'s partition has $N > 2$ elements, and $f$ has a certainty equivalent. Then by A1 and the induction hypothesis, all acts $g'$ measurable with respect to coarser partitions than $f$, and satisfying $\bar{c} \succ g'(\omega) \succ \und{c}$ for all $\omega$, have certainty equivalents. I now show that it is without loss to assume that $f$ is increasing with respect to the same order as $g$, i.e. $f(T_{i+1}) \succ f(T_i)$ for all $i$. To see this, it suffices to show that if $f(T_a) \succ f(T_b) \succ f(T_c)$ then there is an act $f'$ with the same partition that has a certainty equivalent and such that $f'(T_a) \succ f'(T_c) \succ f'(T_b)$ (a symmetric argument shows that there is $f''$ such that $f''(T_{b}) \succ f''(T_c) \succ f''(T_a)$). Define $h$ by $h(T_a) = h(T_c) = f(T_a)$ and $h = f$ on $\Omega\setminus(T_a\cup T_c)$. Then $h$ has a certainty equivalent by the induction hypothesis. Moreover, by continuity there exists $\alpha \in (0,1)$ such that $\alpha h(T_c) + (1-\alpha) f(T_c) \succ f(T_b)$. Moreover, by A2, $f' = \alpha h + (1-\alpha)f$ has a certainty equivalent, as desired. So from now on, assume WLOG that $f$ and $g$ are comonotone. 

By Claim 1, it is without loss to consider $g$ such that $f(T_N) \succ g(T_N)$ and $g(T_1) \succ f(T_1)$. Then there exists $\lambda \in (0,1)$ such that $\lambda f(T_N) + (1- \lambda)f(T_1) \sim g(T_N)$. Define $h_N$ by $h_N = f(T_1)$ on $T_N$ and $h_N = f$ elsewhere. Then $h_N$ has a certainty equivalent by the induction hypothesis. Let $f'_N = \lambda f + (1- \lambda)h_N$. Then $f'_N$ has a certainty equivalent by A2. Moreover, by Monotonicity, $f_N$ defined as $f_N = g(T_N)$ on $T_N$ and $f_N = f'_N$ elsewhere also has a certainty equivalent. Then there also exists $c_{N-1} \in \{f_N(T_N), f_N(T_1) \}$ and $\lambda' \in (0,1)$ such that $\lambda' f_N(T_{N-1}) + (1- \lambda')c_{N-1} \sim g(T_{N-1})$. Define $h_{N-1}$ by $h_{N-1} = c_{N-1}$ on $T_{N-1}$ and $h_{N-1} = f_N$ elsewhere. Then by the induction hypothesis $h_{N-1}$ has a certainty equivalent, and by A2 $f'_{N-1} = \lambda' f_N + (1 -\lambda')h_{N-1}$ also has a certainty equivalent. As before, define $f_{N-1}$ by $f_{N-1} = g(T_{N-1})$ on $T_{N-1}$ and $f_{N-1} = f'_{N-1}$ elsewhere. Then $f_{N-1}$ also has a certainty equivalent. Proceeding in this way, we arrive at an act $f_1 =g$ which has a certainty equivalent, as desired.  

I now need to address the assumption that there exist constant acts $\und{c}, \bar{c}$ such that $\bar{c} \succ g(\omega) \succ \und{c}$ for all $\omega$. Suppose that we have established existence of a certainty equivalent for all acts $f$, measurable with respect to a given partition, that satisfy this interiority assumption. Suppose $g$ is such that $g(T_N) \succsim c \succsim g(T_1) \ \forall \ c \in F_c$ (recall the ordering of $T_i$). Let $g' = \frac{1}{2}g + \frac{1}{2}\left(\frac{1}{2}g(T_N) + \frac{1}{2}g(T_1)\right)$. Then $g'$ satisfies the interiority assumption, and so has a certainty equivalent by hypothesis, which can be written as $\lambda g(T_1) + (1-\lambda) g(T_N)$, where $\lambda \in \left[\frac{1}{4}, \frac{3}{4} \right]$ by Monotonicity. The existence of a certainty equivalent for $g$ will follow from Weak C-Independence if we can show that there exists a $\kappa \in [0,1]$ such that $\frac{1}{2}(\kappa g(T_N) + (1-\kappa)g(T_1)) + \frac{1}{2}\left(\frac{1}{2}g(T_N) + \frac{1}{2}g(T_1)\right) \sim \lambda g(T_1) + (1-\lambda) g(T_N)$. Equating coefficients, this holds for $\kappa = 2 \lambda - \frac{1}{2}$, which is well defined since $\lambda \in \left[\frac{1}{4}, \frac{3}{4} \right]$.
\end{proof}

\subsection{Theorem \ref{thm:simplicity_char}}
\noindent Theorem \ref{thm:simplicity_char} is immediate from \Cref{lem3.1}, A3, A4, and Proposition \ref{prop3.4}.

\begin{proposition}\label{lem3.1}
For any two $N$-element partitions $\tau = \{T_j\}_{j=1}^N, \tau' = \{T'_j\}_{j=1}^N$ of $\Omega$, there is a finite sequence of $N$-element partitions $\{\tau^i\}_{i=1}^K$, starting with $\tau$ and ending with $\tau'$, such that between $\tau^i$ and $\tau^{i+1}$, $N-2$ of the partition cells remain unchanged.
\end{proposition}
\begin{proof}
Let $\tau^i = \{T_j^i\}_{j=1}^N$ be the last partition in the sequence constructed thus far (we begin with $\tau^1 = \tau$). Assume $\tau^i \neq \tau'$, otherwise we are done. The there exists $T' \in \tau'$ such that $T' \cap T^i_1$ and $T' \cap T_2^i$ are both non-empty for some $T_1^i,T_2^i \in \tau^i$. There are two cases to consider. In \textit{Case 1} we can choose $T',T_1^i, T_2^i$ such that $T^i_1 \cup T^i_2 \neq T'$; in \textit{Case 2} we cannot.  

Consider first Case 1. To generate the next partition in the sequence, fix all elements of $\tau^i$ other than $T_1^i, T_2^i$. Choose some sets $r,r' \in \tau^i \vee \tau'$ with $r \subseteq T_1^i, r' \subseteq T_2^i$ such that $r \cup r' \subseteq T'$.\footnote{$\tau \vee \tau'$ is the coarsest common refinement.} Such sets exist since $T' \cap T_1^i$ and $T' \cap T_2^i$ are both non-empty by assumption. If $r \neq T_1^i$ then define $\tau^{i+1}$ to be the modification of $\tau^i$ in which $r$ is merged with $T_2^i$, and all other elements are the same; to be precise, $\tau^{i+1} := \{T_1^i\setminus r, T_2^i\cup r, T_3^i, \dots, T_N^i \}$. If $r = T_1^i$ then define $\tau^{i+1} := \{r\cup r', (T_1^i \cup T_2^i)\setminus(r \cup r' ), T_3^i, \dots, T_N^i \}$. 

Consider now Case 2. We have $T_1^i \cup T_2^i = T'$. Then there exist $T^i_3 \in \tau^i$ and $T_3' \in \tau'$ such that $T_3' \subset T_3^i$. Then define $\tau^{i+1} := \{T_1^i, T_2^i \cup T_3', T_3^i\setminus T_3', T_4^i,\dots,T_N^i\}$ and $\tau^{i+2} := \{T_1^i \cup T_2^i, T_3',T^i_3\setminus T_3', T^i_4,\dots, T^i_N \} = \{T', T_3', T^i_3\setminus T_3', T^i_4,\dots, T^i_N \}$.

As long as $\tau^i \neq \tau'$ the algorithm above delivers $\tau^{i+1} \neq \tau^i$ in Case 1, or $\tau^{i+2} \neq \tau^i$ in Case 2. Moreover, $\tau^k$ is a coarsening of $\tau \vee \tau'$ for all $k$. Since there are at most $N^2$ elements of $\tau \vee \tau'$, the set of possible coarsenings is finite. Thus the algorithm eventually delivers $\tau'$ as long as there are no cycles. But cycles cannot occur. To see this, first notice that if $T \in \tau^i$ and $T \in \tau'$ then $T \in \tau^{i+k}$ for all $k \geq 1$. Therefore the algorithm can only arrive at Case 2 finitely many times along any sequence, since each time it does so it delivers $\tau^{i+2}$ which has an additional 2 cells in common with $\tau'$. Between the steps at which Case 2 is reached there can be no cycles; once $r$ and $r'$ are merged into the same cell in a Case 1 step they are never divided.  
\end{proof}

\subsection{Proof of Theorem \ref{thm3.1} and \texorpdfstring{\Cref{thm:general_rep}}{}}

I prove \Cref{thm:general_rep}, \Cref{thm3.1} is an immediate corollary.

By Theorem \ref{prop3.4}, we know that $F_{\mathcal{T}} \subseteq F_{CE}$ for some downward direct set of partitions $\mathcal{T}$ ($F_{\mathcal{T}}$ be the set of acts measurable with respect to some partition in $\mathcal{T}$). By A0, we can choose such $\mathcal{T}$ which contains all binary partitions.  The first part of the proof shows that preferences on $F_{\mathcal{T}}$ have an expected utility representation under the Basic Conditions and S-Independence. Then Uniform Comparability implies the desired representation.

\noindent \textit{Part 1.}  The restriction of $\succsim$ on $L$ satisfies the von Neumann-Morgernstern axioms, and so is represented by $v(l) = E_l[u]$. Without loss of generality assume that $v(L) \supset [-1,1]$. Fix a partition $\tau = \{ T_i \}_{i=1}^N \in \mathcal{T}$. The restriction of $\succsim$ to acts measurable with respect to $\tau$ satisfies the standard SEU axioms (see for example \cite{gilboa1989maxmin}). Therefore there exists a probability measure $P^{\tau}$ on $\tau$ such that for all $\tau$-measurable acts $f,g$, $f \succsim g$ iff $E_{P^{\tau}}[v\circ f] \geq E_{P^{\tau}}[v\circ g]$. 

It remains to show that there exists a probability $P$ on $\Sigma$ such that $P^{\tau}(A) = P(A)$ for all $\tau \in \mathcal{T}$ and $A \in \Sigma$. For any $A \in \Sigma$ and any $\tau',\tau'' \in F_{\mathcal{T}}$ such that $A \in \tau'$, $A\in \tau''$, it must be that $P^{\tau'} = P^{\tau''}$, since the certainty equivalent for the act $\mathbbm{1}_{A}$ must be the same regardless of which of $P^{\tau'}, P^{\tau''}$ is used to represent preferences. Thus $P(A) = \{P^{\tau}(A): \tau \in \mathcal{T} \}$ is a well defined function. $P$ is non-negative since each $P^{\tau}$ is. 

To show additivity of $P$, it suffices to consider acts in $F_2$. I show now how additivity of $P$ is implied by S-Independence. To see this, it is sufficient to consider $2$-element partitions and $\Omega = \{\omega_1, \omega_2, \omega_3\}$. Let $\tau^1 = \{\{\omega_1 \}, \{\omega_2, \omega_3\}\}$, $\tau^2 = \{\{\omega_1,\omega_3 \}, \{\omega_2\}\}$, and $\tau^3 = \{\{\omega_1,\omega_2 \}, \{\omega_3\} \}$. Let $\mathbbm{1}_A$ be an act such that $v\circ f(\omega) = 1$ for $\omega \in A$, and $0$ otherwise. Then $P^{\tau^i}(\omega_i) = E_{P^{\tau^i}}[\mathbbm{1}_{\{\omega_i\}}]$ for $i \in \{1,2,3\}$. We wish to show that $P$ defined as $P(\omega_i) = P^{\tau^i}(\omega_i)$ for $i \in \{1,2,3\}$ is a well-defined probability on $\Omega$ that is consistent with $P^{\tau^1}, P^{\tau^2}$, and $P^{\tau^3}$. To see this, let $f_1,f_3 \in F_2$ be acts such that $v\circ f_1(\{\omega_2, \omega_3\}) = 1$ and $ v\circ f_1({\omega_1}) = -1$; and $v\circ f_3(\{\omega_1, \omega_2\}) = 1$ and $ v\circ f_3({\omega_3}) = -1$. Then $\frac{1}{2} f_1 + \frac{1}{2} f_3 = \mathbbm{1}_{\omega_2}$, so $E_{P^{\tau^2}}[\frac{1}{2} f_1 + \frac{1}{2} f_3] = P^{\tau^2}(\omega_2)$. But by $S$-Independence we also know that $\frac{1}{2} f_1 + \frac{1}{2} f_3 \sim \frac{1}{2} U(f_1) + \frac{1}{2}U(f_3)$, where $U(f_i)$ is a constant act such that $v(U(f_i)) = E_{P^{\tau^i}}[f_i]$. Then, we also have, with the usual abuse of notation, that $\frac{1}{2} f_1 + \frac{1}{2} f_3 \sim \frac{1}{2}(P^{\tau^1}(\{\omega_2,\omega_3\}) - P^{\tau^1}(\omega_1) + P^{\tau^3}(\{\omega_1,\omega_2\}) - P^{\tau^3}(\omega_3) )$. Thus $\frac{1}{2}(P^{\tau^1}(\{\omega_2,\omega_3\}) - P^{\tau^1}(\omega_1) + P^{\tau^3}(\{\omega_1,\omega_2\}) - P^{\tau^3}(\omega_3) )= P^{\tau^2}(\omega_2)$. Using this equality it is easy to see that $P$ is well defined and consistent; for example $P(\omega_2) = P^{\tau^2}(\omega_2) = 1 - P^{\tau^1}(\omega_1) - P^{\tau^3}(\omega_3) = 1 - P(\{\omega_1,\omega_3\})$. 

By exactly the same argument, we can show that $P$ as defined above is consistent with $P^{\tau'}$ and $P^{\tau''}$ for any $2$-element partitions $\tau'',\tau'$. This gives the desired representation of $\succsim$ for $N$-simple acts. 

\noindent \textit{Part 2.} 
That Uniform Comparability implies the representation in Theorem 1 is immediate; when $\siminf_{\mathcal{T},f}$ and $\simsup_{\mathcal{T},g}$ both exist with $\siminf_{\mathcal{T},f} \succsim \simsup_{\mathcal{T},g}$, they play the role of $h$ and $k$ in the axiom respectively. 

I now show that $\siminf_{\mathcal{T},f}$ must be non-empty for all $f$. The proof for $\simsup_{\mathcal{T},f}$ is analogous. Suppose $\siminf_{\mathcal{T},f}$ is empty. Since $f$ is bounded, there exist constant acts $\bar{f},\und{f} \in F_c$ such that $\bar{f} \succsim f \succsim \und{f}$. Let 
\begin{equation*}
s = \sup_{h \in \{h \in F_{\mathcal{T}} : \  f \geq^0 h\}} \int_{\Omega} E_{h}u \ dP(\omega)
\end{equation*}
By Uniform Comparability, there exists an $h \in \{h \in F_{\mathcal{T}} : \  f \geq^0 h\}$ such that $h \succsim \und{f}$. Choose $\lambda^*$ such that $\lambda^* \int_{\Omega} E_{\bar{f}} u dP(\omega) + (1-\lambda^*) \int_{\Omega} E_{\und{f}} u dP(\omega) = s$. Since, by hypothesis, there are no $h \in F_{\mathcal{T}}$ with $f \geq^0 h$ that have expected utility equal to $s$, Uniform Comparability implies $\neg(f \succsim \lambda^* \bar{f} + (1-\lambda^*) \und{f})$. But this violates Archimedian Continuity. 

\subsection{Proof of Theorem \ref{thm4.1}}

Take $f \not \in F_{CE}$. I first show that $f \sim' \siminf_{N,f}$. Suppose $\exists \  c \in F_c$ such that $c \succ f$ (otherwise the claim is trivial). By consistency $f \succsim' \siminf_{N,f}$. Suppose $\siminf_{N,f} \not \succsim' f$. By continuity there exist constant acts $c_1 \succ c_2$ and $\lambda^* \in [0,1)$ such that $\lambda^* c_1 + (1-\lambda^*)c_2 \sim \siminf_{N,f}$. Then for all $\lambda > \lambda^*$ it follows from Theorem \ref{thm3.1} that $f \not \succsim \lambda c_1 + (1-\lambda) c_2$. Caution then implies that $\lambda c_1 + (1-\lambda) c_2 \succsim' f$. By continuity of $\succsim'$, $\lambda^* c_1 + (1-\lambda^*) c_2 \succsim' f$. The claim follows. Strong consistency for simple acts implies that the restrictions of $\succsim$ and $\succsim'$ to $F_{CE}$ have the same representation. The theorem follows.

\subsection{Existence of \texorpdfstring{$\simsup$ and $\siminf$}{}}\label{sec:existenceproofs}

Rather than look for functions defined by partitions on $\Omega$, I define a dual problem in terms of partitions of $w(\Omega)$. I show that $\simsup_{N,w,P}$ and $\siminf_{N,w,P}$ can be mapped to increasing functions on $w(\Omega)$, and then exploit this monotonicity and fact that $w(\Omega)$ is a bounded interval of $\mathbb{R}$. Put another way, $\Omega$ inherits both an order and a topology from $\mathbb{R}$ and the measurable function $w$, which greatly simplifies the problem of finding simple bounds.

\subsubsection{``Lebesgue Approach''}\label{sec3.1}

Let $\tau(\Omega)$ be the set of all partitions of $\Omega$, and $\tau^N(\Omega)$ the set of $N$-element partitions. When $\Omega$ is a partially ordered set, say that $\tau = \{T_i \}_{i=1}^N$ is an \textit{interval partition} if $T_i$ is an interval for all $i$.\footnote{By an interval in a partially ordered space $(J,\geq)$ I mean a set $I \subseteq J$ such that for all $x,y \in I$ and all $z \in J$ such that $x \geq z \geq y$, $z \in I$. I do not define an interval to be closed, as is sometimes done.}

Fix $w \in B(\Omega)$. For any $h \in B_N(\Omega)$ and let $\tau_h = \{ T_i\}_{i=1}^N$ be $h$'s partition. Define $T'_i := \{r \in w(\Omega) : w(\omega) = r \text{ for some } \omega \in T_i\}$. Define $\tau_h' := \{ T_i'\}_{i=1}^N$ as the cover of $w(\Omega)$ \textit{induced} by $h$. Say that $h$ \textit{induces an interval partition} of $w(\Omega)$ if $\tau_h'$ is a partition of $w(\Omega)$ and $T_i'$ is an interval for all $i$. 

Let $Q$ be the law of $w$, defined by $Q(A) = P(w^{-1}(A))$ for any Borel set $A$. Let $S$ be $Q$'s support. Since $w$ is measurable, $Q$ is a Borel measure on $\mathbb{R}$, and hence a Radon measure. Thus $Q(A) = 0$ for any $A \in \Omega \setminus S$ (see \cite{parthasarathy2005probability}, ch 2). The idea behind the proof is to look at $N$-simple functions on $S$, rather than on $\Omega$. To do this, I need to show that it is possible to move between $B_N(\Omega)$ and $B_N(S)$. I will focus on the existence of $\siminf_{N,w,P}$, as the argument for $\simsup$ is exactly symmetric. 

The following lemma shows that $\siminf_{N,w,P}$, if it exists, will live in the subset of $B_N(\Omega)$ that induce an interval partitions of $S$. The lemma has a symmetric counterpart for $\simsup_{N,w,P}$. 

\begin{lemma}\label{lem3.2}
For any $w \in B(\Omega)$, let $h \in B_N(\Omega)$ be such that $w \geq^0 h$. If $h$ does not induce an interval partition of $S$ then there exists an $N$-simple function $\hat{h}$ that does, and such that $w \geq^0 \hat{h}$ and $E_P[\Hat{h}] \geq E_P[h]$.
\end{lemma}
\begin{proof}
Let $A \subseteq \Omega$ be the set of states $\omega$ such that $h(\omega)> w(\omega)$. We can restrict attention to functions $h$  such that $P(A) = 0$, as this is a requirement for $\siminf_{N,w,P}$. To begin, assume that $A\cap w^{-1}(S) = \varnothing$, so $w \geq_{w^{-1}(S)} h$ (i.e. $w(\omega) \geq h(\omega)$ for all $\omega$ such that $w(\omega) \in S$).

Suppose that $\tau_h'$ is not a partition of $S$. For any $r \in S$ and any $T_i$, $T_j\in \tau_h$ such that $T_i\cap w^{-1}(r) \neq \varnothing$ and $T_j\cap w^{-1}(r) \neq \varnothing$, define $\Hat{h}$ as $\Hat{h}(\omega) = \max\{h(T_i), h(T_j)\} \ \forall \ \omega \in w^{-1}(r)$, and $\Hat{h} = h$ elsewhere. Then $E_P[\Hat{h}] \geq E_P[h]$, and $\tau_{\Hat{h}}'$ will be a partition of $w(\Omega)$. For any such $r$, $w \geq_{w^{-1}(S)} h$ implies $r \geq \max\{h(T_i), h(T_j)\}$. Therefore $w \geq_{w^{-1}(S)} \hat{h}$, and since $P(A) = 0$, $w \geq^0 \hat{h}$

If there are elements of $\tau_h'$ that are not intervals then there exist states $\omega_1, \omega_2 \in \Omega$ with $w(\omega_1) < w(\omega_2)$ and $h(\omega_1) > h(\omega_2)$. Then define $\Hat{h}$ such that $\Hat{h}(\omega_2) = \Hat{h}(\omega_1)= h(\omega_1)$, and $h = \Hat{h}$ elsewhere. Clearly $E_P[\Hat{h}] \geq E_P[h]$ and $w \geq^0 \Hat{h}$.

Now, I want to show that it is without loss to assume $A\cap w^{-1}(S) = \varnothing$. For any $\omega \in A$, if there exists $r \in h(\Omega \setminus A)$ such that $w(\omega) \geq r$ then we can replace $h(\omega)$ with $r$ without altering the value of $h$, or the fact that it is $N$-simple. Assume therefore that $w(\omega)< \min h(\Omega\setminus A)$ for all $\omega \in A$. Since $w \geq_{\Omega\setminus A} h$, this implies that $w(\omega) < \inf w(\Omega \setminus A)$ for all $\omega \in A$. But then $w(A) \cap S = \varnothing$, or equivalently $A\cap w^{-1}(S) = \varnothing$.
\end{proof}

\begin{corollary}\label{cor3.1}
When $\siminf_{N,w,P}$ and $\simsup_{N,w,P}$ are non-empty, they contain functions that induce interval partitions of $S$.
\end{corollary}

For any function $\tilde{h} \in B_N(w(\Omega))$ we can define a function $h \in B_N(\Omega)$ by $h(\omega) = \tilde{h}(w(\omega))$. Moreover, $E_Q[\tilde{h}] = E_P[h]$ by the definition of $Q$. Lemma \ref{lem3.2} and Corollary \ref{cor3.1} are useful because they allow us to do the converse: given a function $h \in B_N(\Omega)$ that induces a partition of $w(\Omega)$ we can define a function $\tilde{h} \in B(w(\Omega))$ by $\tilde{h}(r) = h(w^{-1}(r)) $. By $h(w^{-1}(r))$ I mean the value taken by $h$ for all $\omega \in w^{-1}(r)$. For $\tilde{h}$ to be well defined it is therefore necessary that $h$ induce a partition of $w(\Omega)$. Interval partitions are easy to work with, as we will see, since they can be described by a vector of $N$ cut-offs.

\subsubsection{Existence proof}

\begin{proposition}\label{prop7}
Suppose $\Omega$ is a closed interval in $\mathbb{R}$, $P$ has full support, and $w \in B(\Omega)$ is continuous, and increasing. Then $\siminf_{N,w,P}$ is non-empty. 
\end{proposition}
\begin{proof}
Since $P$ has full support and $w$ is continuous, it is without loss to assume that $w \geq \siminf_{N,w,P}$. To find $\siminf_{N,w,P}$, it will be sufficient to restrict attention to increasing functions with interval partitions. This follows since for any non-interval partition $T$ there exist states $\omega_1,\omega_2 \in \Omega$ satisfying $w(\omega_1) \leq w(\omega_2)$ and $h_T(\omega_1) > h_T(\omega_2)$. Then adding $\omega_2$ to the partition cell that contains $\omega_1$ leads to a new $N$-simple function $h'$ with $h' \leq w$ and $E_P[h'] \geq E_P[h_{\tau}]$. When working with interval partitions label the cells so that $T_{i} > T_{i-1}$ (in the obvious order). Moreover, since $w$ is increasing we can assume that each interval in $\tau$ contains its left endpoint (since it is always optimal to assign a state to the higher interval). Any such partition can by fully described by a vector $z_{\tau} = (t_1,\dots,t_{N-1})$ of $N-1$ cutoffs. Given a cutoff vector $z$, let $h_z$ be the function constructed as above using the partition described by $z$. Since the set of such cutoff vectors is compact, it only remains to show that $z \mapsto E_P[h_z]$ is upper semi-continuous. This follows from right-continuity of the CDF and continuity of $w$. 
\end{proof}

\begin{proof} (\Cref{thm3.2}).
By Corollary \ref{cor3.1}, it is without loss to look $\siminf_{N,w,P}$ in the subset of $B_N(\Omega)$ that induce interval partitions of $S$. It is therefore without loss to look for functions $\tilde{h} \in B_N(S)$ that have interval partitions, i.e. we solve $\max_{\hat{h}\in B_N(S): w \geq_S \hat{h}} E_Q[\hat{h}]$. This problem has a solution by Proposition \ref{prop7}. Let $\tilde{h}$ be the solution, and define $h = \tilde{h}\circ w$ on $S$. Let $h = c$ on $\Omega \setminus S$, for some $c \in h(S)$. Then $h$ will be $N$-simple and satisfy $w \geq^0 h$. Moreover, Corollary \ref{cor3.1} implies that $h \in \siminf_{N,w,P}$.
\end{proof}

\subsection{Proof of Proposition \ref{prop:sim_char}}
\begin{proof}
I give the proof for $\siminf$; $\simsup$ is symmetric. If $f^{-1}(A)$ is non-null for every open neighborhood of $\inf E_fu(\Omega)$ then there must be some $\omega$ such that $\inf E_fu(\Omega) \geq E_{\siminf_{N,f}}u(\omega)$. If there was a violation of statewise dominance for some $\omega'$ then it could be removed by specifying $E_{\siminf_{N,f}}u(\omega') = E_{\siminf_{N,f}}u(\omega)$. If the condition doesn't hold then $\exists \ \omega$ such that it is strictly sub-optimal to have $E_fu(\omega) \geq E_{\siminf_{N,f}}u(\omega)$. 
\end{proof}

\subsection{Proof of Proposition \ref{prop3.2}}
\begin{proof}
For any act $g$ with partition $\tau$ and any set $A$ define the random variable $\Hat{\e}_{\chi}[g|A]$ as follows:
\begin{equation*}
    \Hat{\e}_{\chi}[g|A] =
    \begin{cases}
        \frac{1}{\Hat{P}_{\chi}(A)}\sum_{i=1}^N f(T_i)\frac{1}{K}\sum_{x \in \chi} \mathbbm{1}\{x \in T_i\cap A \} \ \ \ \ \text{if }
        \Hat{P}_{\chi}(A) > 0\\
        \e[g|A] \ \ \ \ \text{otherwise}
    \end{cases}
\end{equation*}
This definition is necessary since we will be dealing with finite samples, so $\Hat{P}_{\chi}(A) = 0$ with positive probability even when $P(a) > 0$. 

Let $f_b$ be the act equal to $b$ on $T'$ and $f$ elsewhere. In what follows $T'^c = \Omega-T'$. Notice that $\varepsilon_f(\chi)$ can be written as 
\begin{align*}
    \varepsilon_f(\chi)&= \Hat{\e}_{\chi}[f|T'](1 - \Hat{P}_{\chi}(T'^c)) + \Hat{\e}_{\chi}[f|T'^c]\Hat{P}_{\chi}(T'^c)- E[f] \\
    & = \Hat{\e}_{\chi}[f|T'] - \e[f|T'] +\left(\Hat{\e}_{\chi}[f|T'^c] - \Hat{\e}_{\chi}[f|T']\right)\Hat{P}_{\chi}(T'^c) \\
    & \hspace{5em} - \left(\e[f|T'^c] - \e[f|T']\right)P(T'^c)
\end{align*}
Similarly 
\begin{equation*}
    \varepsilon_{f_b}(\chi) = \left(\Hat{\e}_{\chi}[f|T'^c] - b\right)\Hat{P}_{\chi}(T'^c) - \left(\e[f|T'^c] - b\right)P(T'^c).
\end{equation*}
Thus 
\begin{equation*}
\begin{split}
    &\varepsilon_f(\chi) =\\
    &\varepsilon_{f_b}(\chi) + \underbrace{\Hat{\e}_{\chi}[f|T'] - \e[f|T'] + \left(b - \Hat{\e}_{\chi}[f|T']\right)\Hat{P}_{\chi}(T'^c) + \left(\e[f|T'] - b\right)P(T'^c)}_{\xi(\chi)}
\end{split}
\end{equation*}
The weak inequality in the definition of SOSD will follow by Jensen's inequality if I can show that $\e[\xi | \{\chi : \varepsilon_{f_b}(\chi) = m \}] = 0$ for all $m$ in the range of $\varepsilon_{f_b}$, where the expectation is taken with respect to the measure on datasets induced by $P$ and the i.i.d. sampling procedure. The strict inequality will follow since the distribution of $\xi$ is non-degenerate. 

Notice that, because sampling is i.i.d., $\e \left[\Hat{\e}_{\chi}[f|T'] \ \Big| \ \Hat{\e}_{\chi}[f|T'^c], \Hat{P}_{\chi}(T'^c) \right] = \e[f| T']$, where we use here the specification of $\Hat{\e}_{\chi}[f|T'] = \e[f|T']$ when $\Hat{P}_{\chi}(T'^c) = 1$. Independence between $\Hat{\e}_{\chi}[f|T'], \Hat{\e}_{\chi}[f|T'^c]$, and $\Hat{P}_{\chi}(T'^c)$ does not hold for finite samples due to integer restrictions. Fortunately all that we need is the stated conditional mean independence condition. Given this condition
\begin{align*}
    \e\left[ \xi | \varepsilon_{f_b} \right] &= \e\left[\e\left[ \xi | \Hat{P}_{\chi}(T'^c) \right] \big| \varepsilon_{f_b}\right]\\
    &= \e\left[ (b - \e[f|T'])(\Hat{P}_{\chi}(T'^c) - P(T'^c)) \big| \varepsilon_{f_b} \right]
\end{align*} 
where the first equality is just the law of iterated expectations and the second follows from conditional mean independence. Clearly if $b = E[f|T']$ the expectation is zero, so we are done. 
\end{proof}

\subsection{Proof of Proposition \ref{prop:smoothsampling}}
\begin{proof}
A1 has already been shown. A2 follows immediately from Jensen's inequality. A3 follows from the fact that $\phi$ is strictly concave and $\hat{G}_{\bar{f}} >_{SOSD} \hat{G}_{f}$; for any 2-element partition $\tau''$ of $T_1 \cup T_2$ we can choose a binary act $b$ on $\tau''$ that is arbitrarily close to a constant and such that $\hat{E}_{\chi}[b|\tau''] = \hat{E}_{\chi}[f|\tau'']$. So $b$ can be chosen such that the act $f''$ that is equal to $b$ on $\tau''$ and $f$ elsewhere will satisfy $E_{\hat{G}_{f''}}\left[\phi\left(\hat{E}[f''] \right)\right] > E_{\hat{G}_{f}}\left[\phi\left(\hat{E}[f] \right)\right]$. A0 is satisfied because for any event $A$ we can choose a bet on $A$ that is arbitrarily close to constant. A4 holds since null events are those for which there is not data, and when using the bootstrap the value on such sets is irrelevant. 
\end{proof}

\subsection{Proofs for the applications}

\subsection{Proof of Proposition \ref{prop9.1}}

\begin{proof}\textit{Proposition \ref{prop9.1}(i)}.
In what follows $N$ and $x$, the level of savings, will be fixed. Assume without loss of generality that the $\siminf$ is unique for any $\alpha$ (the same proof applies for any selection). Let $P^{\alpha}$ be the perceived distribution (induced by the cutoffs in (\ref{eq32})) when the proportion of savings allocated to the risky asset is $\alpha$. For a fixed level of savings $x$ the allocation problem becomes.
\begin{equation}\label{eq9.3}
\max_{\alpha \in [0,1]} \int u((1-\alpha)xR_b + \alpha x r) dP^{\alpha}(r).
\end{equation}

Notice that for the unconstrained DM, for whom $P^{\alpha}$ is replaced with $P$, the objective is concave in $\alpha$. This means that the derivative of the objective crosses zero (at most once) from above. The result will follow if I can show that at any $\alpha$ the derivative of the objective in (\ref{eq9.3}) with respect to $\alpha$ is greater for the unconstrained DM then for the constrained one. Recall that $\siminf_{N,f}$ (equivalently $P^{\alpha}$) is defined by a maximization problem, which satisfies the conditions of the envelope theorem \citep{milgrom2002envelope}. Then, by the envelope theorem, the derivative with respect to $\alpha$ of the objective function in (\ref{eq9.3}) is given by 
$$
\int u'((1-\alpha)xR_b + \alpha x r)(xr - xR_b) dP^{\alpha}(r). 
$$
Recall that $P \succsim_{FOSD} P^{\alpha}$ for all $\alpha$. It is therefore sufficient to show that $u'((1-\alpha)xR_b + \alpha x r)(xr - xR_b)$ is increasing in $r$. To do this I show that for any $b,s \geq 0$
\begin{equation*}
    \dfrac{d}{dr} [u'(R_bb + rs)R_b] \leq \dfrac{d}{dr} [u'(R_bb + rs)r],
\end{equation*}
or equivalently
\begin{equation}\label{eq9.5}
 (r - R_b)s \leq -\dfrac{u'(R_bb + rs)}{u''(R_bb + rs)} 
\end{equation}
Notice that by assumption $R_b \in (\und{R}_s, \bar{R}_s)$, so that the left hand side of (\ref{eq9.5}) is negative for $r$ low enough, while the the right hand side is always strictly positive. So it is sufficient to show that the derivative with respect to $r$ of the LHS of (\ref{eq9.5}) is less than that of the RHS for all $r$, i.e.
\begin{equation*}
     \dfrac{d}{dr}[(r - R_b)s] \leq \dfrac{d}{dr} \left[ -\dfrac{u'(R_bb + rs)}{u''(R_bb + rs)} \right] \Longleftrightarrow \dfrac{u'''(R_bb + rs)u'(R_bb + rs)}{u''(R_bb + rs)^2} \geq 2
\end{equation*}
This holds for all $r$, for example, for CRRA utility $u(x) = x^{1-\gamma}/(1-\gamma)$ when $\gamma \geq 1$. 
\end{proof}

\begin{proof}\textit{Proposition \ref{prop9.1}(ii)}
To prove the proposition define 
\begin{equation*}
V(x,N) = u(w - x) + \beta \max_{\alpha} \int u((1-\alpha)xR_B + \alpha x r)dP^{\alpha, x}(r).
\end{equation*}
where $P^{\alpha,x}$ is the lower-perceived distribution corresponding to the $\siminf$ for the act induced by $x,\alpha$. Define
\begin{equation*}
V(x,\infty) = u(w - x) + \beta \max_{\alpha} \int u((1-\alpha)xR_B + \alpha x r)dP(r).
\end{equation*}
I will show that $V$ has decreasing differences, in the sense that for  $x'' > x'$
\begin{equation*}
V(x', \infty) - V(x'', \infty) \geq V(x', N) - V(x'', N),
\end{equation*}
A sufficient condition for the above inequality is that for all $x \in [x',x'']$
\begin{equation}\label{eq9.6}
\begin{split}
      \dfrac{d}{dx}&\left[ \max_{\alpha} \int u((1-\alpha)xR_B + \alpha x r)dP^{\alpha,x}(r) \right] \\
      &\geq \dfrac{d}{dx}\left[ \max_{\alpha} \int u((1-\alpha)xR_B + \alpha x r)dP(r) \right]  
\end{split}
\end{equation}
Let $\alpha(x,N)$, $\alpha(x,\infty)$ be the optimal allocation proportions. Let $\xi(\alpha,r) = (1-\alpha)xR_B + \alpha x r$. I prove that the inequality in (\ref{eq9.6}) holds in two parts. First 
\begin{align}
    &\dfrac{d}{dx}\left[ \max_{\alpha} \int u(\xi(\alpha,r))dP^{\alpha, x}(r) \right]\\  
    &= \int u'(\xi(\alpha(x,N),r))((1-\alpha(x,N))R_b + \alpha(x,N)r)  dP^{\alpha, x}(r) \\
    & \geq  \int u'(\xi(\alpha(x,N),r))((1-\alpha(x,N))R_b + \alpha(x,N)r)  dP(r). \label{eq9.8}
\end{align}
where the first equality follows from the envelope theorem. The inequality in (\ref{eq9.8}) will follow by $P \geq_{FOSD} P^{\alpha,x}$ if the integrand in (\ref{eq9.8}) is decreasing in $r$. Taking the derivative and rearranging we can see that this is the case if and only if the coefficient of relative risk aversion $-u''(z)z/u'(z)$ is greater than or equal to $1$. 

To complete the proof that (\ref{eq9.6}) holds I show that
\begin{align*}
    \int u'(&\xi(\alpha(x,N),r))((1-\alpha(x,N))R_b + \alpha(x,N)r)  dP(r) \\
    &\geq \dfrac{d}{dx}\left[ \max_{\alpha} \int u((1-\alpha)xR_B + \alpha x r)dP(r) \right].
\end{align*}
This will follow by the envelope theorem and Proposition \ref{prop9.1}(i), which says that $\alpha(x,\infty) \geq \alpha(x,N)$. We need only show that for $\alpha \in [\alpha(x,N), \alpha(x,\infty)]$ the expression on the LHS of the above inequality is decreasing in $\alpha$. To be precise, we need that 
$$
\int u'((1-\alpha)xR_B + \alpha x r)((1-\alpha)R_b + \alpha r) dP(r)
$$
is decreasing in $\alpha$. Taking the derivative, we need
\begin{equation}\label{eq9.9}
\int [u''(\xi(\alpha,r))((1-\alpha)xR_b + \alpha xr) + u'(\xi(\alpha,r))](r-R_b) dP(r) \leq 0.
\end{equation}
Notice that the term in brackets in the integrand, $u''((1-\alpha)xR_B + \alpha x r)((1-\alpha)xR_b + \alpha xr) + u'((1-\alpha)xR_B + \alpha x r)$, is less than or equal to zero by the assumption of $-u''(z)z/u'(z) \geq 1$, but $(r-R_b)$ is negative for low values of $r$ and positive for high values. Assuming CRRA utility, (\ref{eq9.9}) reduces to
\begin{equation}\label{eq9.12}
\int u'((1-\alpha)xR_B + \alpha x r)(xr - xR_b) dP(r) \geq 0 
\end{equation}
Notice that the LHS of (\ref{eq9.12}) is exactly the derivative with respect to $\alpha$ of expected utility. Thus for all $\alpha \leq \alpha(x,\infty)$, (\ref{eq9.12}) holds by concavity of the objective function. Then (\ref{eq9.6}) holds, as desired. 
\end{proof}

\subsection{Proof of Proposition \ref{prop9.3}}
\begin{proof}
Let $P_N^{b,s}$ be the induced distribution for the agent with capacity $N$ (either Cautious or Reckless). The first order conditions for the constrained agent's problem are given by 
\begin{align*}
    [b] \ \ \ &u'(w - b - ps) = \dfrac{d}{db}\left[\beta \int u(w + bR_b - sr) dP_N^{b,s}(r)\right] \\
    [s] \ \ \ &u'(w - b - ps)p = \dfrac{d}{ds}\left[\beta \int u(w + bR_b - sr) dP_N^{b,s}(r) \right].
\end{align*}
In equilibrium (with $b = s = 0$) the FONC for $b$ implies that $R_B = 1/\beta$, regardless of capacity or attitude. Let $p^N$ be the equilibrium risky asset price for the capacity $N$ agent. Zero net supply requires that
\begin{equation}\label{eq9.13}
    u'(w)p^{N} = \dfrac{d}{ds}\left[ \beta \int u(w + sr) dP_N^{0,s}(r) \right]_{s=0}.
\end{equation}
Clearly when $s = 0$ we have $\beta \int u(w + sr) dP_{N''}^{0,s}(r) = \beta \int u(w + sr) dP_{N'}^{0,s}(r)$ for all $N',N''$. For any $s \neq 0$ for a cautious (reckless) DM, $\beta \int u(w + sr) dP_{N'}^{0,s}(r) <  (>) \  \beta \int u(w + sr) dP_{N''}^{0,s}((r)$ for $N'' > N'$, since $P_{N}^{0,s}$ is the solution to a maximization (minimization) problem. Therefore the derivative on the RHS of (\ref{eq9.13}) is increasing (decreasing) in $N$ when the DM is cautious (reckless). The result follows.
\end{proof}

\subsection{Proofs for Section \ref{sec:principal_agent}}

\subsubsection{Proof of Lemma \ref{lem:reckless_complex}}
\begin{proof}
Fix a wage schedule $\tilde{w}$ which induces effort level $a$. Let $\tilde{u} = u \circ \simsup(\tilde{w}, P_a)$ and $\tilde{w}(\bar{\omega}) = \sup_{\omega \in \Omega} \tilde{w}(\omega)$ and $\tilde{u}(\bar{\omega}) = u(\tilde{w}(\bar{\omega}))$. Assume that $\simsup_{N,\tilde{w},P_a}$ is singleton (this is without loss; we could do everything in terms of selections from $\simsup_{N,\tilde{w},P_a}$). Order the states so that $\tilde{w}$ is increasing.

Suppose there is an $N$-simple contract $w$ that induces effort level $a$. Then it will be possible to find a cell $T_i$ of $w$'s partition and a sub-interval at the lower end of $T_i$ such that reducing the payoff slightly on this sub-interval, leaving the contract otherwise unchanged, does not change the $\simsup$ under $P_a$. In other words, if $w'$ is the contract so obtained then $\simsup_{N,w,P_a} = \simsup_{N,w',P_a}$. Moreover, since $w'$ dominated statwise by $w$, we have $U(\simsup_{N,w,P_{a'}},P_{a'}) \geq U(\simsup_{N,w',P_{a'}}, P_{a'})$ for all $a'$, so that $w'$ also induces effort level $a$.   
\end{proof}

\subsubsection{Proof of Proposition \ref{prop6.3}}
\begin{proof}

Assume that $\tilde{w}$ is simple or continuous. We want to show that the principal can improve on a simple contract by offering a discrete jump at the top. Let $a$ be the optimal effort under $\tilde{w}$. Let $\tau = (t_i )_{i=0}^{N+1}$ be a vector of cut-offs where $t_0 = 0$, $t_{N+1} = 1$. Let $T^N$ be the set of such cut-off vectors. For any $\tau \in T^N$, we can a define an $N$-simple function $f_{\tau}$ that uniformly dominates $\tilde{u}(\omega) := u(\tilde{w}(\omega),a)$ as $f_{\tau}(\omega) = \tilde{u}(t_i)$ for $\omega \in (t_{i-1}, t_{i}]$ and $i \in \{1, \dots, N+1 \}$. Clearly any element of $\simsup_{N,\tilde{u},P_a}$ can be represented as $f_{\tau}$ for some $\tau \in T^N$. 

\textit{Claim 1.} There exists $\delta > 0$ and $\kappa > 0$ such that $E_{P_a}[f_{\tau}] - E_{P_a}[\simsup_{N,\tilde{u},P_a}] > \delta$ for all $\tau \in T^N$ such that $t_N \geq 1-\kappa$. 

\textit{Proof of Claim 1.} Since $\tilde{u}$ is continuous in a neighborhood of $1$, and $P_a$ is absolutely continuous, we have $\tau \mapsto E_{P_a}[f_{\tau}]$ continuous as well. Moreover, $E_{P_a}[\simsup_{N-1,\tilde{u},P_a}] > E_{P_a}[\simsup_{N,\tilde{u},P_a}]$. The claim follows, since $\tau \in T^N$ and $t_N = 1$ implies $\tau \in T^{N-1}$.
 
Given Claim 1, we can define the desired modification of $\tilde{w}$, denoted by $\hat{w}$, as follows. Let $\varepsilon < \kappa$. For $\omega \in (1-\kappa, 1-\varepsilon)$, let $\hat{w}(\omega) = \tilde{w}(\omega) - \beta \sigma(\omega)$ for any $\sigma > 0$ such that $\tilde{w}(\omega) - \beta \sigma(\omega)$ is increasing. Let $w = \tilde{w}$ for all other output levels. Denote $\hat{u}(\omega) := u(w(\omega),a)$ Then by Claim 1, for $\beta$ small enough we have $E_{P_a}[f_{\tau}] - E_{P_a}[\simsup_{N,\hat{u},P_a}] > 0$. Thus $\simsup_{N,\hat{u},P_a} = \simsup_{N,\tilde{u},P_a}$. Moreover, since $\tilde{w}$ uniformly dominates $\hat{w}$, $\hat{w}$ will also induce effort level $a$ (as in the proof of Proposition \ref{prop6.1}). 
\end{proof}

\subsection{Proofs for Section \ref{sec8}}

Let $f'_{-}(\omega)$ and $f'_+(\omega)$ be left and right derivatives of $f$ at $\omega$ respectively. The proof of Proposition \ref{prop8.1} makes use of the following lemma. 
\begin{lemma}\label{lem8.1}
Let $f$ be a decreasing and continuous function. Suppose $f$ has a kink at $d$ ($f_+'(d) > f_-'(d)$). Then $\siminf(f)$ is constant in a neighborhood of $d$.
\end{lemma}
\begin{proof}

Suppose that there is a cut-off at $d$. That is, $t_n = d$ and $\siminf(f)$ is discontinuous at $t_n$. Fix $t_{n-1}, t_{n+1}$. Denote the value generated by an cutoff in $(t_{n-1}, t_{n+1})$ by 
\begin{equation*}
    V(t) = [P(t)-P(t_{n-1})]u(f(t)) + [P(t_{n+1}) - P(t)]u(f(t_{n+1}))
\end{equation*}
Then optimality of $t_n =d$ implies that the left derivative of $V$ at $d$ must be positive and the right derivative must be negative:
\begin{equation}\label{eq8.1}
    V_-'(d) = [P(d)-P(t_{n-1})]u'(f(d))f_-'(d) + p(d)[u(f(d)) - u(f(t_{n+1}))] \geq 0
\end{equation}
\begin{equation}\label{eq8.2}
    V_+'(d) = [P(d)-P(t_{n-1})]u'(f(d))f_+'(d) + p(d)[u(f(d)) - u(f(t_{n+1}))] \leq 0
\end{equation}
Equations (\ref{eq8.1}) and (\ref{eq8.2}) imply
\begin{equation*}
    f_+'(d) \leq \dfrac{-p(d)}{u'(f(d))[P(d)-P(t_{n-1})]}[u(f(d)) - u(f(t_{n+1}))] \leq f_-'(d)
\end{equation*}
which contradicts $f_+'(d) > f_-'(d)$. \end{proof}

Let $t_N$ be the highest cut-off defining $\siminf(y|d)$. Clearly $t_N \leq d$, since $y$ is flat above $d$ for a full insurance contract. Moreover Lemma \ref{lem8.1} implies that $t_N < d$. 

\begin{proof}
\textit{Proposition \ref{prop8.1}.}
Let $l(\omega|d,c)$ be the amount paid by the consumer when the loss is $\omega$ given a contract with deductible $d$ and coverage rate $c$. Denote the perceived value to a cautious agent of an insurance contract characterized by $d,c$ as 
\begin{equation}\label{eq8.3}
    U^N(d,c) = \max_{\hat{t}_1,\dots,\hat{t}_N} \sum_{n=1}^{N+1} [P(\hat{t}_n) - P(\hat{t}_{n-1})]u(w - l(\hat{t}_n|d,c))
\end{equation}
A fully rational agent would value the contract at
\begin{equation*}
    U^{\infty}(d,c) = \int_0^d u(w - \omega)dP(\omega) + \int_d^{\Bar{\omega}} u(w - d - (1-c)(\omega -d))dP(\omega)
\end{equation*}

Let $\{ t_1, \dots, t_N\}$ be the solution to the maximization problem in \ref{eq8.3},  and let $n^* = \max\{n \in \{1, \dots, N\}  : t_n \leq d\}$ be the index of the highest cut-off below $d$. By the envelope theorem 
\begin{equation*}
    U^N_d(d,c) = -c\sum_{n = n^*+1}^{N+1} [P(t_n) - P(t_{n-1})]u'(w - d -(1-c)(t_n-d)).
\end{equation*}
Moreover $U^{\infty}_d(d,c) = -c \int_{d}^{\Bar{\omega}} u'(w - d -(1-c)(\omega-d))dP(\omega)$. Notice that
\begin{align*}
   U^N_d(d,c) &< c[P(d) - P(t_{n^*})]u'(w - d -(1-c)(t_{n^* + 1}-d)) \\
   & \ \ \ \ \ \ \ -c\sum_{n = n^*+1}^{N+1} [P(t_n) - P(t_{n-1})]u'(w - d -(1-c)(t_n-d))\\
   & \leq -c \int_{d}^{\Bar{\omega}} u'(w - d -(1-c)(\omega-d))dP(\omega) \\
   & = U^{\infty}_d(d,c) 
\end{align*}
where the first inequality follows since Lemma \ref{lem8.1} implies $t_{n^*} < d$, and the second from concavity of $u$. Notice that the second inequality holds with equality if and only if $c = 1$ for $u$ strictly concave. This proves part 1 of Proposition \ref{prop8.1}.

The proof of part 2 is similar. In this case
\begin{align*}
    U^N_c(d,c) &= \sum_{n = n^*+1}^{N+1} [P(t_n) - P(t_{n-1})]u'(w - d -(1-c)(t_n-d))(t_n - d) \\
    & > - [P(d) - P(t_{n^*})]u'(w - d -(1-c)(t_n-d))(t_n - d) \\
    & \ \ \ \ \ \ \ +\sum_{n = n^*+1}^{N+1} [P(t_n) - P(t_{n-1})]u'(w - d -(1-c)(t_n-d))(t_n - d) \\
    &\geq \int_{d}^{\Bar{\omega}} u'(w - d -(1-c)(\omega-d))(\omega-d)dP(\omega) \\
    &= U^{\infty}_c(d,c)
\end{align*}
where again the first inequality follows from Lemma \ref{lem8.1} and the second from concavity of $u$. 
\end{proof}

\subsubsection{Proof of Proposition \ref{prop8.2}}
This result follows from a similar argument as Proposition \ref{prop8.1}. The key to the proof of Proposition \ref{prop8.1} was the result of Lemma \ref{lem8.1} that $t_{n^*} < d$. When the baseline contract is full insurance above $d$ the following result allows us to draw  the analogous conclusion that the value of the greatest cut-off below $d$ is increasing in $N$.

Let $t^N = t_0^N, \dots, t_{N+1}^N $ be the cutoffs defining $\siminf(y|d)$ when the agent has capacity $N$ and $t^{N+1} = t_0^{N+1}, \dots, t_{N+2}^{N+1}$ be the cutoffs for capacity $N+1$. When $c = 1$ we have the immediate corollary of Lemma \ref{lemA.14}, which says that the largest cut-off below $d$ is increasing in $N$. Notice that for $c=1$, $n^* = N$ where $N$ is the capacity of the agent. This follows since ex-post wealth is constant above $d$, and so it would not be optimal to have a cut-off above $d$.    
\begin{corollary}\label{cor8.1}
For a full insurance contract above a deductible, $t_N^N < t_{N+1}^{N+1}$.
\end{corollary}

We can now prove Proposition \ref{prop8.2}.
\begin{proof}
\textit{(Proposition \ref{prop8.2})} For $c = 1$ an envelope theorem implies that
\begin{align*}
    U^N_d(d,c) &= -[1 - P(t^N_N)]u'(w - d) \\
    & \leq [P(t^{N+1}_{N+1}) - P(t^N_N)]u'(w - d) -[1 - P(t^N_N)]u'(w - d) \\
    & = U^{N+1}_d(d,c)
\end{align*}
The proof for decreasing $c$ (recall that $c$ is bounded above by 1), is similar. 
\end{proof}

\subsubsection{Propositions \ref{cor8.2} and \ref{cor8.3}}
Propositions \ref{cor8.2} and \ref{cor8.3} follow immediately from Corollary \ref{cor5.1} and \ref{corA.12} respectively.

\subsubsection{Proof of Propositions \ref{prop8.5} and \ref{prop8.6}}
I give the proof Proposition \ref{prop8.5} here. The argument for Proposition \ref{prop8.6} is essentially identical. 
\begin{proof}
Let $m$ be the out of pocket maximum, and $U^N(m)$ be the perceived value, holding $c$ and $d$ fixed. Let $\und{l}(m) = \min\{\omega \in \Omega : l(\omega|d,c,m) = m \}$. Let $C^*(N,[0,\und{l}(m)])$ be the cut-offs corresponding to the elements of the plan $\siminf$. Since $P$ is absolutely continuous and $\und{l}(\cdot)$ is continuous, $C^*(N,[0,\cdot])$ is upper-hemicontinuous by Berge's maximum theorem. 

Recall that for a decreasing and continuous function $f$ on an interval of the reals the cell function corresponding to simple lower bounds is given by
\begin{equation}\label{eq6.10}
    v([a,b]) = f(b) P(\{\omega \in [a,b] \}).  
\end{equation}

Since utility is strictly decreasing on $[0,\und{l}(m)]$ and $P$ is full support the cell function generating $U^N(m)$ is strictly submodular by Lemma \ref{lem1}. Moreover, it satisfies the conditions for regularity in \cite{tian2015optimal}, so by \cite{tian2015optimal} Theorem 3, $C''$ and $C'$ are sandwiched for all $C'' \in C^*(N+1,[0,\und{l}(m)])$ and $C' \in C^*(N,[0,\und{l}(m)])$, so upper-hemicontinuity implies that for $\varepsilon$ small enough there exist sandwiched selections from $C^*(N,m)$ and $C^*(N+1, m-\varepsilon)$. 

For an interval $I \subseteq \Omega$ and $C \in C^N(I)$ let $V(I,C)$ be the coarse value corresponding to the cell function in (\ref{eq6.10}). Since utility is constant above $\und{l}(m)$, we can write  $U^N(m) = V([0,\und{l}(m)],C^*(N,[0,\und{l}(m)])) + (1 - P(\und{l}(m)))u(w - m)$ (with an abuse of notation when $C^*(N,[0,\und{l}(m)])$ is non-singleton). The result follows from Lemma \ref{lemA.13}.
\end{proof}

\subsubsection{Lemma \ref{lem8.4}}

\begin{lemma}\label{lem8.4}
Consider two plans with no out-of-pocket maximum, coverage rate $c$, deductibles $d$ and $d'$, with $d < d'$, and premiums $p,p'$ such that $p' - p = (1 - c)(d' - d)$. Then a cautious agent is indifferent between the two plans if the lowest cut-off defining $\siminf$ of the high deductible plan is (weakly) greater than $d'$. \footnote{The $\siminf$ for the high deductible plan need not be unique. I state the result in this way for simplicity, but it holds as long as the condition is satisfied for a selection from the $\siminf$. In fact, if $\siminf$ for the high deductible plan is single valued then the converse holds as well. }
\end{lemma}

\begin{proof}
Let $t$ be the lowest cut-off for the $\siminf$ of the high deductible plan. Since the high deductible dominates the low deductible plan, any set of cut-offs defines a less preferred lower bound for the latter than for the former. Since the $\siminf$ of the high deductible plan is also dominated statwise by the low deductible plan it is an element of the $\siminf$ for the low deductible plan as well. 
\end{proof}


\bibliography{cu_citations}

\newpage

\section{Appendix: For online publication}\label{sec:online}

\subsection{A Note on Transitivity}
The Simplicity Conditions highlight the critical role of transitivity in the characterization. Transitivity is often regarded as a weak rationality condition. I maintain Transitivity for the usual reasons (no money pumps, normative desirability, etc.). It is worth pointing out, however, that it also rules out many non-separable models of pairwise choice. For example, consider the following model: acts $f$ and $g$ are comparable if and only if either \textit{i}) the join (coarsest common refinement) of their partitions has less than $N$ elements, or \textit{ii}) one statewise dominates the other. These preferences are consistent with Monotonicity, Archimedian Continuity, Nontriviality and the Simplicity Conditions, but not Transitivity. The idea behind this model is that the decision maker must evaluate the difference between the payoffs from the two acts. Such preferences are ruled out by Transitivity however, precisely because Transitivity implies that any two acts which are comparable to all constant acts can be compared. 

I extend preferences from the set of well-understood acts to complex acts using only the statewise dominance order, which satisfies Transitivity. Assuming Transitivity of $\succsim$ rules out extensions that violate Transitivity, although some of these do have intuitive appeal. Relaxing transitivity while accommodating incompleteness arising from subjective complexity is an interesting topic for future work.   

\subsection{MEU in insurance}\label{sec:MEUinsurance}

The insurance plans considered here all induce payoffs that are monotone in the the state. If the set $C$ of beliefs over which the DM with MEU preferences minimizes contains a first-order stochastically dominant belief then this will be the minimizing belief regardless of plan characteristics. In this sense the individual will behave exactly like an expected utility maximizer.

In some cases MEU and cautious preferences predict similar behavior. For example, in the MEU framework we can consider individuals who minimize over smaller sets of beliefs, which is analogous to the idea of increasing capacity. Consider comparing the relative value of two plans which induce ex-post wealth functions $f_1, f_2$ for two individuals, one of whom has a larger set of beliefs. If the larger set of beliefs contains an upper bound in the first-order stochastic dominance (FOSD) order, and the payoff difference $f_1 - f_2$ is increasing (decreasing) in the state then the individual with the larger belief set will value $f_1$ more (less) relative to $f_2$ than the individual with the smaller set. For example, the individual with a larger set of beliefs will value reductions in the out-of-pocket maximum more than an individual with a smaller set of beliefs.

\section{Comparative Statics}\label{app:compstat}

Given the centrality of the $\siminf$ and $\simsup$ in the representations of incomplete and complete preferences discussed above, it helpful in applications to understand how these functions vary with the parameters of the problem. This section contains a number of comparative statics results that are helpful in this regard. In addition, some of the results apply to a more general class of problems and are of independent interest. I characterize the responses $\siminf$ and $\simsup$ to changes in both the capacity and the beliefs. I also show that the marginal returns to additional capacity are higher when the relevant state space is larger. These results are used in the applications. 

I will discuss real valued acts, with the understanding that all conclusions apply to the utility images of any acts. Moreover, I will make statements about \textit{every} element of $\siminf$ and $\simsup$, with the understanding that these apply ``up to sets of measure zero under $P$''. I will first focus on properties of $\siminf_{N,w,P}$ when $\Omega = [\und{\omega}, \bar{\omega}] \subset \mathbb{R}$ and $w$ is increasing. Results in this setting can be extended in two ways. First, since $\siminf_{N,w,P} = - \simsup_{N,-w,P}$ all results regarding $\siminf$ can be translated directly to $\simsup$. Second, using the ``Lebesgue approach'' of Section \ref{sec3.1}, results for increasing functions on an interval can be translated to results for arbitrary bounded functions on an arbitrary state space. Finally, for simplicity I will assume throughout that $f \geq \siminf_{N,w,P}$. Using Proposition \ref{prop:sim_char} we know when this will hold. Moreover, Proposition \ref{prop:sim_char} tells us where violations of statewise dominance can occur. The results presented below will apply with only minor modifications when $\neg(f \geq \siminf_{N,w,P})$, as we can just ignore the zero measure set on which violations of statewise dominance occur. 
 
\subsection{Submodular cell functions}

I first present some results in a general setting. The framework, as well as many of the general results mentioned here, are from \cite{tian2015optimal} and \cite{tian2016monotone}. Let $\Omega = [\und{\omega}, \bar{\omega}] \subset \mathbb{R}$.\footnote{It matters that $\Omega$ is bounded, but the fact that it is closed here is irrelevant} Let $I(\Omega)$ be the set of interval subsets of $\Omega$. A \textit{cell function} is a function $v: I(\Omega) \rightarrow \mathbb{R}$. An \textit{interval partition} is a partition consisting only of intervals. Given an interval partition $\tau = \{T_i \}_{i=1}^N$ of $\Omega$, the \textit{coarse value} associated with cell function $v$ is defined as 
\begin{equation*}
    V(\tau) = \sum_{i=1}^N v(T_i).
\end{equation*}
I will be interested in cell functions that are submodular. 
\begin{definition}
A cell function $v$ is (strictly) \textbf{submodular} if for all intervals $I,I'$ with $I \cap I' \neq \varnothing$, we have
\begin{equation*}
 v(I \cap I') + v(I \cup I') \  (<) \leq \ v(I) + v(I').  
\end{equation*}
\end{definition}

The following Lemma says that elements of $\siminf$ and $\simsup$ will have interval partitions. The proof follows from that given for Lemma \ref{lem3.2}.

\begin{lemma}
If $w$ is increasing then every element of $\siminf_{N,w,P}$ and $\simsup_{N,w,P}$ has an interval partition. 
\end{lemma}

For an interval $I$ let $ \und{w}_a(I) = \inf_{\omega \in I}\{w_a(\omega)\}$. Given this Lemma, we know that when $f \geq \siminf_{N,u(a),P}$, $\siminf_{N,u(a),P}$ will be defined by the partition maximizing the coarse value associated with the cell function, which I will call the \textit{$\siminf$ cell function} $v(I) = \und{w}_a(I) P(I)$. Moreover, I will show that this coarse value is submodular. I turn now to general properties of submodular cell functions. 

Submodularity can be thought of as diminishing returns to larger intervals. Indeed, this intuition is made concrete by the following observation from \cite{tian2015optimal}. For each cell $I = [a, b]$ and $c$ in $I$, let $\Delta v(c, I) = [v([a, c]) + v([c, b])] - v(I)$. Notice that for submodular cell functions $v(c,I)$ is always non-negative. We can make the following observation.

\begin{proposition} \label{prop201}
\textbf{(\cite{tian2015optimal}, Observation 1)} A cell function is (strictly) submodular if and only if for all intervals $I,I'$ with ($I' \subsetneq I$ ) $I' \subseteq I$ and any $c\in I'$ we have 
$$
\Delta v(c,I) \ (>) \geq \ \Delta v(c,I'). 
$$
\end{proposition}

Suppose $\tau \in T^N(X)$ is an interval partition of the sub-interval $X \in I(\Omega)$. Then $\tau$ can be described by a set of cut-offs $C = \{t_i \}_{i=1}^{N-1} \in \mathbb{R}^{N-1}$ giving the interior endpoints of the partition intervals. I will sometimes use the convention $t_0 = \inf\{X\}$ and $t_{N} = \sup\{X\}$ when refering to a partition of the interval $X$. Cut-off states can be assigned to partition cells in any way. \footnote{Technically, a full description of the partition should include, in addition to the cut-offs, a vector in $\{0,1\}^{N+1}$ identifying whether each cut-off state is assigned to the interval immediately above or below it. A cut-off assigned the the lower interval is considered to be lower than the same cut-off assigned to the higher interval, otherwise the usual order on $\mathbb{R}^{N}$ is used to order cut-off vectors.} Since for finding the $\siminf$ it will always be optimal to assign cut-off states to the higher interval I will assume this assignment from now on. Let $C^N(\Omega)$ be the set of cut-offs defining interval partitions of $\Omega$. Abusing notation, I will write $V(\Omega,C)$ to denote the coarse value of the partition induced on $\Omega$ by cut-off vector $C \in C^N(\Omega)$. 

Endow $C^N(\Omega)$ with the pointwise partial order.  Denote the least upper bound and greatest lower bound of two cut-off vectors $C, C' \in C^N(\Omega)$ by $C \vee C'$ and $C \wedge C'$. Define the union of two sets of cut-offs in the obvious way, as the ordered union of the two sets. \footnote{By this definition the union of two cut-off vectors defines a partition which is the join, in the refinement sense, of the partitions defined by each of the individual cut-off vectors. This is different than the join of the cut-off vectors, which I define as the coordinate-wise maximum. The later notion is restricted to cut-off vectors of the same length, whereas the union can be taken of any two cut-off vectors.}

\begin{definition}
A coarse value is (strictly) \textbf{supermodular} if it is (strictly) supermodular as a function of cut-off vectors. 
\end{definition}

The following result also relates submodularity of cell functions to diminishing returns of cell division. It is extremely useful in the comparative statics results. 

\begin{proposition}\label{prop202}
\textbf{(\cite{tian2015optimal}, Observation 3)} The coarse value is (strictly) supermodular if and only if the cell function is (strictly) submodular. 
\end{proposition}

Propositions \ref{prop201} and \ref{prop202} are from \cite{tian2015optimal}, Observations 1 and 3. Although \cite{tian2015optimal} does not consider the strict version of either result, this follows by the same argument.

Let $W(N,\Omega) = \sup_{C \in C^N(\Omega)} V(\Omega,C)$. \cite{tian2015optimal} also proves that supermodular coarse values exhibit diminishing marginal returns to capacity. 
\begin{proposition}
\textbf{(\cite{tian2015optimal}, Theorem 1)}. For a supermodular coarse value $W(N+1,\Omega) - W(N, \Omega) \leq W(N, \Omega) - W(N-1, \Omega) $.
\end{proposition}

It is easy to verify that the $\siminf$ cell function satisfies the condition of Proposition \ref{prop201}, and is thus submodular, so that the associated coarse value is supermodular by Proposition \ref{prop202}. 

\begin{lemma} \label{lem1}
The $\siminf$ cell function $v(\cdot)$ is submodular. It is strictly submodular if $P$ has full support and $w_a$ is strictly increasing. 
\end{lemma}
\begin{proof}
For any interval $[a,b]$ and $c \in [a,b]$ we have
\begin{align*}
   \Delta v(I,c) &= \und{w}_a([a,c])P([a,c]) + \und{w}_a(c,b)P([c,b]) - \und{w}_a([a,b])P([a,b]) \\
   & = (\und{w}_a([c,b]) - \und{w}_a([a,c]))P([c,b])
\end{align*}
where the last equality follows since $\und{w}_a([a,c]) = \und{w}_a([a,b])$ when $w_a$ is increasing, as it is here by assumption. 
Suppose $[a,b] \subseteq [l,m]$. Then $\und{w}_a([c,b]) = \und{w}_a([c,m])$, $\und{w}_a([l,c]) \leq \und{w}_a([a,c])$ and $P([c,b]) \leq P([c,m])$, so $\Delta v([a,b],c) \leq \Delta v([l,m],c)$. The claim follows from Proposition \ref{prop201}.
\end{proof}

\begin{corollary}\label{cor1}
The coarse value for the $\siminf$ cell function is supermodular. It is strictly supermodular if $w_a$ is strictly increasing and $P$ has full support. 
\end{corollary}

Proposition \ref{prop201} relates submodularity of the cell function to the returns to making a subdivision of a cell at a given place. The coarse values of submodular cell functions also satisfy the following increasing differences property.

\begin{lemma}\label{lemA.14}
If the cell function is submodular then 
\begin{itemize}
    \item For all $a \in \Omega$, $b'' \geq b' \geq a$ and $C'' = \{t_i''\}_{i=1}^{N-1}, C' = \{t_i'\}_{i=1}^{N-1}\in C^N([a,b'])$ with $t_{N-1}'' \geq t_{N-1}'$, 
    \begin{equation}\label{eq13.2}
        V([a,b''], C'') - V([a,b''], C') \geq V([a,b'], C'') - V([a,b'], C').
    \end{equation}
    \item For all $b \in \Omega$, $b \geq a'' \geq a'$ and $C'' = \{t_i''\}_{i=1}^{N-1}, C' = \{t_i'\}_{i=1}^{N-1}\in C^N([a'',b])$ with $t_1'' \geq t_1'$,
    \begin{equation}\label{eq13.3}
        V([a',b], C') - V([a',b], C'') \geq V([a'',b], C') - V([a'',b], C'').
    \end{equation}
\end{itemize}
Moreover, if the cell function is strictly submodular then the inequality in (\ref{eq13.2})  is strict when $t_{N-1}'' > t_{N-1}'$ and $b'' > b'$, and that in (\ref{eq13.3}) is strict when $t_1'' > t_1'$ and $a'' > a'$. 
\end{lemma}
\begin{proof}
I prove the first claim. The proof of the second is analogous. The key is that the partitions induced on $[a,b'']$ and $[a,b']$ by a cut-off vector $C$ differ only in the highest cell of the partition. This implies that $V([a,b''], C'') - V([a,b''], C'' \cup C') = V([a,b'], C'') - V([a,b'], C'' \cup C')$, since $t_{N-1}'' \geq t_{N-1}'$ implies that for a given interval the highest cells of the partitions induced by $C''$ and $C'' \cup C'$ are the same. It also means that $V([a,b''], C'' \cup C') - V([a,b''], C') \geq V([a,b'], C'' \cup C') - V([a,b'], C')$ by the characterization of submodular cell functions in Proposition \ref{prop201}. Therefore
\begin{align*}
    &V([a,b''], C'') - V([a,b''], C') \\ 
    &= V([a,b''], C'') - V([a,b''], C'' \cup C') + V([a,b''], C'' \cup C') - V([a,b''], C') \\
    & =  V([a,b'], C'') - V([a,b'], C'' \cup C') + V([a,b''], C'' \cup C') - V([a,b''], C') \\
    & \geq V([a,b'], C'') - V([a,b'], C'' \cup C') + V([a,b'], C'' \cup C') - V([a,b'], C') \\
    & = V([a,b'], C'') - V([a,b'], C')
\end{align*}
\end{proof}

The following standard comparative statics result follows from Lemma \ref{lemA.14} and Proposition \ref{prop202}. For sets $A,B$ I write $A \geq_{SSO} B$ when $A$ is larger than $B$ in the strong set order.

\begin{corollary}\label{corA.9}
With a submodular cell function, for intervals $I'', I' \subseteq \Omega$ with $I'' \geq_{SSO} I'$,  $C^*(N,I'') \geq_{SSO} C^*(N,I')$.
\end{corollary}

\begin{corollary}\label{corA.10}
With a strictly submodular cell function the elements of $C^*(N,X)$ are ordered for all $N$ and all $X \subseteq \Omega$. 
\end{corollary}

Related to the first part of Corollary \ref{corA.10}, \cite{tian2015optimal} provides alternative conditions under which the elements of $C^*(N,X)$ are ordered. 

\subsubsection{Changing N}\label{sec13.1.1}
It will often be useful to understand how the cut-offs defining optimal partitions change with capacity. The following definition is from \cite{tian2015optimal}.

\begin{definition}
Two sets of cut-offs $C \in \mathbb{R}^N$ and $C' \in \mathbb{R}^{N+1}$ are \textbf{sandwiched} if $t'_i  < t_{i+1}  <t'_{i+1}$ for all $i$.  
\end{definition}

Say that cut-offs are weakly sandwiched if the strict inequalities in the above definition are replaced with weak inequalities. The following result is related to Proposition 2 in \cite{tian2015optimal}, which shows that the maximal elements of $C^*(N,[a,b])$ and $C^*(N+1, [a,b])$, when they exist, are sandwiched (and similarly for the minimal elements). While \cite{tian2015optimal} uses continuity of the coarse value to guarantee that minimal and maximal elements of the sets of optimal cut-off vectors exist, we can conclude that sandwiched selections exist without this assumption.

\begin{corollary}\label{corA.12}
With a sub-modular cell function, if $C^*(N,[a,b])$ and $C^*(N+1,[a,b])$ are non-empty then for any $C' \in C^*(N,[a,b])$ there exists $C'' \in C^*(N+1,[a,b])$ such that $C'$ and $C''$ are weakly sandwiched.
\end{corollary}

\begin{proof}
Start with $C'' = \{t_i''\}_{i = 1}^N \in C^*(N+1,[a,b])$ and $C' = \{t_i' \}_{i=1}^{N-1} \in C^*(N,[a,b])$. Define $\tilde{C} = \{\tilde{t}_i\}_{i=1}^N \equiv (\{t_i''\}_{i = 1}^{N-1} \wedge C')\cup t_N''$ and $\hat{C} = \{\hat{t}_i\}_{i=1}^N \equiv (\{\tilde{t}_i\}_{i=2}^{N}\vee C') \cup \tilde{t}_1$. By Lemma 5 in \cite{tian2015optimal}, $\hat{C} \in C^*(N+1,[a,b])$.

If $\{ \hat{t}_i\}_{i=1}^{N-1} \leq C'$ then it will follow that $\hat{C}$ and $C'$ are weakly sandwiched. To see that this holds, note that $\hat{t}_1 = \tilde{t}_1 = \min\{t_1'',t_1'\} \leq t_1'$ and for $i > 1$, $\hat{t}_i = \max\{\tilde{t}_i, t'_{i-1} \}$. Moreover, $\tilde{t}_N = t_N''$ and $\tilde{t}_i = \min\{ t_i'', t_i' \}$ for $i<N$. Thus for $1 < i <N$, we have $\hat{t}_i = \max\{\min\{t_i'',t_i' \}, t_{i-1}' \} \leq t_i'$. 

\end{proof}

Under an additional condition on the cell function, \cite{tian2015optimal} shows that all cut-offs are sandwiched. 

\vspace{3mm}
\noindent\textbf{\cite{tian2015optimal}, Theorem 3.} With a \textit{regular} cell function, $C''$ and $C'$ are sandwiched for all $C'' \in C^*(N+1, [a,b])$ and $C' \in C^*(N,[a,b])$. 
\vspace{3mm}

The definition of a regular cell function can be found in \cite{tian2015optimal}.

\subsubsection{Marginal returns to capacity}
Consider an interval state space. Submodularity of the cell function means that the benefit of dividing an interval at a certain point is higher for larger (in the inclusion order) intervals. The following proposition shows that with a submodular cell function there will also be lower returns to increasing capacity when dividing a sub-interval. Assume the state space is an interval, and let $S \subseteq S'$ be two interval subsets of the state space. 

Let $C = \{t_i\}_{i=1}^{N-1}$ be a cut-off vector defining an $N$-element interval partition of the interval $X \subseteq \Omega$, where $t_0 = \inf\{ X\}$ and $t_N = \sup\{X\}$. I follow the convention that cut-off states are assigned to the higher interval, although this has no bearing on the results. If $X$ is closed at the bottom and open at the top then 
\begin{equation*}
    V(X,C) = \sum_{i=0}^{N-1} v([t_i, t_{i+1})) + v([t_N, t_{N+1}]).
\end{equation*}
Let the set of cut-offs defining $N$-element interval partitions of $X$ be $C^N(X) \subseteq \mathbb{R}^{N-1}$. Let $C^*(N,X) = \argmax_{C \in C^N(X)} V(X,C)$, and assume $C^*(N,X)$ is non-empty for all $X \subseteq \Omega$ (as is the case when the coarse value is the expectation of a simple lower or upper bound of a bounded function). Let $W(N,X) = \max_{C \in C^N(X)} V(C)$. The following lemma is similar to the characterization of submodularity given by Proposition \ref{prop201}. It states that the marginal returns to capacity are higher when dividing larger intervals.

\begin{lemma}\label{lemA.13}
Let the cell function be submodular, and let $S,S'$ with $S \subseteq S'$ be two intervals. Assume that there exist selections from $C^*(N,S')$ and $C^*(N+1,S)$ that are sandwiched. Then $W(N+1,S') - W(N,S') \geq W(N+1,S) - W(N,S)$.
\end{lemma}
\begin{remark}
Recall that there exist sandwiched selections from $C^*(N,X)$ and $C^*(N+1,X)$ for any $X$ by \ref{corA.12} or \cite{tian2015optimal}, Theorem 3. Thus the condition of Lemma \ref{lemA.13} will be satisfied if $v()$ is regular, the optimal cut-offs are continuous in the endpoints of the intervals, and $S,S'$ are sufficiently close. 
\end{remark}

\begin{proof}
Let $C' = \{t'_i\}_{i=1}^{N-1}$ and $C = \{t_i\}_{i=1}^{N}$ be the sandwiched selections from $C^*(N,S')$ and $C^*(N+1,S)$ respectively. Since the cut-offs are sandwiched it must be that $\inf \{S\} = t_0 \leq t_1 \leq t'_1$ and $\sup\{S\} = t_{N+1} \geq t_{N} \geq t'_{N-1}$. Therefore $C'$ also defines an $N$-element partition of $S$. The partition induced by $C'$ on $S$ and that induced by $C'$ on $S'$ differ only on the highest and lowest partition cells, which are larger for the latter. Similarly for the partitions induced by $C$ on $S$ and $S'$. By optimality of $C'$ and $C$ we have 
\begin{align*}
    W(N+1,S') - W(N,S') & = W(N+1,S') - V(S',C') \\
    & \geq V(S',C) - V(S',C'),
\end{align*}
and 
\begin{align*}
    W(N+1,S) - W(N,S) &= V(S,C) - W(N,S) \\
    & \leq V(S,C) - V(S,C')
\end{align*}
so it suffices to show $V(S',C) - V(S',C') \geq V(S,C) - V(S,C')$. We have
\begin{equation*}
    V(S' , C) - V(S', C') = V(S', C') - V(S', C \cup C') + V(S', C \cup C') - V(S', C')
\end{equation*}
and
\begin{equation*}
    V(S, C) - V(S, C') = V(S, C) - V(S, C \cup C') + V(S, C \cup C') - V(S, C'). 
\end{equation*}
By the sandwiched property of $C$ and $C'$, $V(S', C) - V(S', C \cup C') = V(S, C) - V(S, C \cup C')$. Since the partition induced by $C'$ on $S$ and that induced by $C'$ on $S'$ differ only on the highest and lowest partition cells, which are larger for the latter, submodularity of the cell function implies that $V(S',C \cup C') - V(S',C') \geq V(S, C \cup C') - V(S, C')$.
\end{proof}

\subsubsection{Changing the distribution}

I look here at how the functions in $\simsup$ and $\siminf$ change as beliefs change. As before, I will discuss the $\siminf$ case, but all the results hold without alteration for $\simsup$. I remain in the one dimensional setting. The results will be extended in the next section. As I will be varying beliefs while holding capacity fixed I will make explicit the dependence on $P$ of the, cell function, coarse value, and optimal cut-offs by writing these as $v(\cdot|P)$, $V(\Omega, \cdot| P)$, and $C^*(N,\Omega|P)$.

Let $\Hat{P}$ and $P$ be distributions, with densities $\Hat{p}$ and $p$ respectively. Recall that distribution $\Hat{P}$ that dominates $P$ according to the \textit{monotone likelihood ratio} (MLR) property if there exists a non-negative increasing function $\alpha$ such that $\Hat{p}(\omega) = \alpha(\omega)p(\omega)$ for all $\omega \in \Omega$. 

Given two real valued functions $g$ and $h$ defined on a partially ordered set $Z$, say that $g$ dominates $h$ by the \textit{interval dominance order} ($g \succeq_I h$) if $f(z'') \geq f(z') \Longrightarrow g(z'') \geq g(z')$ holds for all $z'' \geq z'$ such that $g(z) \geq g(z') \ \forall \ z \in [z',z''] = \{z \in Z : z' \leq z \leq z'' \}$. The following is implied by \cite{tian2016monotone}, Proposition 3.

\begin{proposition} \label{prop5.1}
Let $\Omega$ be an interval of $\mathbb{R}$, and $G \succeq_{MLR} P$. Assume that $w$ is absolutely continuous (and thus differentiable a.e.). Then $V(\Omega, \cdot|G) \succeq_{I} V(\Omega, \cdot|P)$ on $C^N(\Omega)$ for all $N$. 
\end{proposition}

If $w$ is decreasing then the same conclusion holds; $V^N(\cdot|G) \succeq_I V^N(\cdot|P)$.

By Proposition \ref{prop202} and Lemma \ref{lem1}, $V$ is supermodular in $C$. By Theorem 1 in \cite{quah2009comparative} and Proposition \ref{prop5.1} we obtain the following.
\begin{corollary} \label{cor5.1}
If $G \succeq_{MLR} P$ then $C^*(N, \Omega|G) \geq_{SSO} C^*(N, \Omega|P) $.
\end{corollary}

\end{document}